\documentclass[11pt]{article}

\usepackage{physics}
\usepackage{amsmath}
\usepackage{tikz}
\usepackage{mathdots}
\usepackage{yhmath}
\usepackage{cancel}
\usepackage{color}
\usepackage{siunitx}
\usepackage{array}
\usepackage{multirow}
\usepackage{amssymb}
\usepackage{gensymb}
\usepackage{tabularx}
\usepackage{extarrows}
\usepackage{booktabs}
\usetikzlibrary{fadings}
\usetikzlibrary{patterns}
\usetikzlibrary{shadows.blur}
\usetikzlibrary{shapes}
\usepackage{amsthm}

\newtheorem{theorem}{Theorem}[section]
\newtheorem{lemma}[theorem]{Lemma}
\newtheorem{proposition}[theorem]{Proposition}
\newtheorem{corollary}[theorem]{Corollary}
\newtheorem{definition}[theorem]{Definition}
\newtheorem{remark}[theorem]{Remark}
\usepackage{comment}

\usepackage{subfiles}

\DeclareMathOperator*{\argmax}{arg\,max}

\title{Computing a Best Response against a Maximum Disruption Attack}

\author{\`Alvarez, C. and Messegu\'e, A.}

\begin{document}

\maketitle

\begin{abstract}
Inspired by scenarios where the strategic network design and defense or immunisation are of the central importance,  Goyal et al. \cite{Goyal2} defined a new Network Formation Game with Attack and Immunisation. The authors  showed that despite the presence of attacks,  the game has high social welfare properties and  even though the equilibrium networks can contain cycles, the number of edges is strongly bounded.  Subsequently, Friedrich et al. \cite{Lenzner} provided a polynomial time algorithm for computing a best response strategy for the maximum carnage adversary which tries to kill as
many nodes as possible, and for the random attack adversary, but they left open the problem for the case of maximum disruption adversary.
This adversary  attacks the vulnerable region that minimises the post-attack social welfare.

In this paper we address our  efforts to this question.  We can show  that computing a best response strategy given a player $u$ and  the strategies of all players but $u$, is polynomial time solvable when the initial network  resulting from the given strategies is connected. Our algorithm is based on a dynamic programming and has some reminiscence to the  knapsack-problem, although is considerably more complex and involved.

\end{abstract}

\section{Introduction}

Strategic network formation arises in settings where agents
receive some benefit from being connected to other agents, but  also 
incur costs due to  the creation of  links. We focus our attention on 
the network formation game with attack and immunisation  defined in \cite{Goyal2} which is an extension of the well-known reachability model  introduced by Bala and Goyal in \cite{Goyal1}. This extension 
incorporates  a strong adversary  and immunisation. 
The adversary  attacks and destroys  a node of the network and  then
this attack is spread virus-like to neighboring non-immunised nodes and destroys them as well. Besides deciding the subset of nodes to whom to buy links,   each player has to decide whether she wants to buy immunisation against eventual attacks.

The benefit of a player is defined as  the expected size of her connected component post-attack and  her  cost depends on the number of links bought by the player and the cost of being immunised if it is the case.

In \cite{Goyal2} the authors provide  structural results for their model and raise the  open problem of settling the complexity of computing a best response strategy.

The existence of an efficient best response algorithm
for a network formation game is in general  rare. For  related network formation
models, e.g. \cite{Biloetal, Biloetal2015, Chauhanetal2016, CordLenzner2015, Ehsanietal2015, Fabrikantetal2003, Alvarezetal2016}  where players strive for a central position in the network, it
has been shown that the best response problem is indeed NP-hard.

\medskip
\noindent
{\bf Related results}.
Lenzner et al.  in \cite{Lenzner} showed  that the natural model defined by Goyal et al. in \cite{Goyal2} is one of the few examples of a tractable realistic model for strategic network formation and thereby answer an open question
by these authors. They provided a polynomial time algorithm  for computing a best response
strategy for  the \emph{maximum carnage adversary} which tries to kill as
many nodes as possible, and for the natural variant which employs \emph{random attack adversary}.  Settling the complexity of computing a best response strategy against  \emph{maximum disruption  adversary} was left as an open problem in \cite{Lenzner}.
This adversary  attacks the vulnerable region that minimises the post-attack social welfare. Notice that a naive approach to calculate the best response for a player would consider all $2^n$ possible subsets of links as well as the possibility of to be or not immunised. The algorithm presented in  \cite{Lenzner} circumvents this combinatorial explosion computing first a potential subset of  vulnerable candidates using a dynamic programming  if the node  is not immunised, or  a greedy programming  if the node being immunised. And second, when immunised nodes are considered, the given network is simplified in order to be tractable by  a dynamic programming approach.



\medskip
\noindent
{\bf Our Contribution}.
We assume that the outcome network  of the strategies of all players   but  the one of who wants to compute a best response, is connected. 
That is, the initial network after dropping any strategy for player $u$ is connected. Our main result is to show that in that situation  the  best response problem of $u$  given the strategies of all the rest of players can be computed in polynomial time  for the case of maximum disruption adversary.  

The key ideas  that lead us to prove it are the following:

\begin{enumerate}

\item  The definition of what we call \emph{delta value}. This parameter allow us to characterise the set of nodes attacked by the adversary. Besides this, it is also crucial the  definition of our  \emph{meta-tree}. This tree-like structure   underlying any configuration from the model conducts to  a first simplification on how can we think of a potential best response. 

\item The definition  of \emph{restricted strategy}  and the corresponding \emph{restricted utility}. 
Since we want to compute efficiently a best response for a given player, we partition the set of potential best responses into a collection of mutually disjoint subsets depending on different parameters whose  values are upper bounded  by a polynomial of the number of players. Hence, if  we know how to find a strategy achieving the maximum restricted utility in each subset, then we can select among them the one having the  maximum utility as a best response for the original problem. Moreover, 
we show that the restricted utility in a sub-tree $T$ can be expressed in terms of the restricted utilities of the  strategies restricted to each of the its sub-trees of $T_1,...,T_k$.
We also define  the natural  concept of restricted strategy having the maximum possible restricted utility, named as 
\emph{restricted BR-strategy}. Analogously to the restricted utility, 
we can exploit the structure of the meta-tree.


\item Finally,  two recurrence relations corresponding to the restricted BR-strategies and their respective restricted utilities can be given using the previous characterisations.
\end{enumerate}

Since the parameters of such recurrences can take a polynomial number (in $n$, the number of  players) of possible values taking at most polynomial values (in $n$),   then we can conclude that 
the Best Response problem  is polynomial time computable using a dynamic programming approach.

\medskip
\noindent
{\bf Organisation of the paper}.
In Section 2 we introduce the model  and we provide some extra definitions that will be fundamental to obtain our main result. 
In Section 3 we introduce the  concepts of \emph{restricted strategy} 
and \emph{restricted utility} and we show how to characterise recursively the restricted utility 
exploiting  the structure of the meta-tree, breaking the original problem into easier sub-problems to solve.
In Section 4, we consider   the natural  concept of restricted strategy having the maximum possible restricted utility, which we call \emph{restricted BR-strategy} and we show how to express  the utility of  a restricted BR-strategy recursively. 
Finally, in Section 5 we define two recurrence relations corresponding to the restricted BR-strategies and their respective 
restricted utilities that allow us to solve the Best Response problem in polynomial time using a dynamic programming approach.

\section{Best Response Strategies against a Maximum Disruption Adversary}

We start introducing the model as a strategic game following the definitions given in \cite{Lenzner,Goyal1}. 

First of all,  we consider $N = \left\{ 1,...,n\right\}$ the set of the $n$ players or agents that correspond to the nodes of the network. We use the terms player, node or agent interchangeably. Each player then buy links at price $\alpha >0$ to the other players and can buy immunisation at price $\beta > 0$, where $\alpha$ and $\beta$ are prefixed parameters of the model. We denote by $s_v=(S_v, i_v)$ the \emph{strategy} of player $v$ where $S_v \subseteq N \setminus \left\{v\right\}$, the \emph{link-strategy}, is the subset of players to which node $v$ buys links and $i_v \in \left\{ 0, 1\right\}$, the \emph{immunisation-value}, is the  value of player $v$ that indicates whether $v$ has bought immunisation.  

The \emph{strategy profile}, $\bold{s} = (s_1,...,s_n)$, is obtained considering the strategies of all the players where $s_v = (S_v,i_{v})$ for each $v \in N$. Then any strategy profile $\bold{s}$ gives place to the undirected graph $G(\bold{s}) = (N, \cup_{v \in N} \cup_{w \in S_v}\left\{(v,w) \right\})$.
Notice that the immunisation-values for each player induce a partition of the players into two sets which are $\mathcal{I},\mathcal{U}$ the immunised and the vulnerable set of players, respectively.

Once the network $G(\bold{s})$ is formed an adversary attacks a vulnerable player according to a strategy previously known by the players and such attack spreads through the network reaching all the vulnerable nodes that can be reached with a path of vulnerable nodes from the node that the adversary has attacked. The set of vulnerable nodes that the adversary can attack, the set of \emph{target nodes}, is denoted as $A(\bold{s})$. Then, $CC_v(a,\bold{s})$ is the connected component that can be reached from $v$ after the adversary attacks the player $a$ given that the players have adopted the respective strategies from $\bold{s}$. With this notation then the utility of a player $v$ is defined as $$U_v(\bold{s}) = -\alpha |S_v|-i_v \beta + \frac{1}{|A(\bold{s})|}\sum_{a \in A(\bold{s})} |CC_v(a,\bold{s})|$$ 
In this way we then define the \emph{social welfare} as $U(\bold{s}) = \sum_{v \in N} U_v(\bold{s})$. 

In \cite{Goyal1} the authors propose three distinct kind of adversaries: 

\begin{enumerate}
\item The \emph{maximum carnage} adversary selects with uniform probability any of the largest regions of contiguous vulnerable nodes and then attacks with uniform probability any of its vulnerable nodes.

\item The \emph{random attack} adversary attack with uniform probability any vulnerable node.

\item The \emph{maximum disruption} adversary attacks with uniform probability any vulnerable node from the regions minimising the post-attack social welfare.
\end{enumerate}

In this paper we deal with the maximum disruption adversary. Our aim is to show that  a \emph{best response}  for a given player $u$ from $N$ is  polynomial time computable. Therefore, given the tuple of the strategies of all the players except $u$, noted as $\bold{s}_{-u}$, we want to compute in polynomial time a strategy $(S,i)$ for $u$ such that if $\bold{s}_u = (S,i)$ then $U_u(\bold{s})$ has maximum value among all such possible strategies $(S,i)$.

For this reason it is more convenient to assume that we are given $\bold{s}_{-u}$ and, for any strategy $(S,i)$ for player $u$ we denote as $G(S,i)$ 
the corresponding undirected network after $u$ has adopted strategy $(S,i)$. Then, $\mathcal{U}[G(S,i)]$ and $\mathcal{I}[G(S,i)]$ are the vulnerable and immunised nodes from $G(S,i)$, respectively. 

We also define in a similar way $A(S,i)$, the set of target regions from $G(S,i)$ that the adversary attacks and $CC_z(a,S,i)$ is the connected component containing $z \in V(G(S,i))$ after the adversary attacks $a$ given that the strategy for $u$ is $(S,i)$. Then $U(u,S,i)$ is the utility of $u$ given that the strategy for $u$ is $(S,i)$:
\begin{equation*}U(u,S,i) = -\alpha |S| - i \cdot \beta + \frac{1}{|A(S,i)|}\sum_{a \in A(S,i)}|CC_u(a,S,i)|\end{equation*}

Finally, it is also useful to consider the following definition. For each node $z \in \mathcal{U}(G(S,i))$ we define the \emph{delta value}  of $z$ with respect the strategy $(S,i)$, noted as $\Delta(z,S,i)$, to be the sum of the squares of the size of the distinct connected components that we obtain after removing the node $z$ together with all vulnerable nodes connected to $z$ via a path of vulnerable nodes in $G(S,i)$. More specifically:
\begin{equation*}\Delta(z,S,i) = \sum_{t \in V(G)}|CC_t(z,S,i)|^2\end{equation*}
Then we define:
\begin{equation*}\Delta(S,i) = \min_{z \in \mathcal{U}[G(S,i)]}\Delta(z,S,i)\end{equation*} 

Next we provide the very first result that allow us to understand how the maximum disruption adversary behaves. More precisely, we see that the subset of attacked nodes by the maximum disruption adversary are precisely the vulnerable nodes that minimise their corresponding  delta value:

\begin{proposition}
$A(S,i) = \left\{ a \in \mathcal{U}(G(S,i)) \mid \Delta(a,S,i) = \Delta(S,i) \right\}$
\end{proposition}

\begin{proof} Let $G = G(S,i)$, (or $G= G(\bold{s})$ with $\bold{s}_u=(S,i)$) and let $A \subseteq V(G)$ with $A \neq \emptyset$ be a subset of nodes from $G$. Then:
\begin{equation*} 
\frac{1}{|A|} \sum_{v\in V(G)}\sum_{a\in A}|CC_v(a,S,i)|=
\frac{1}{|A|} \sum_{a \in A} \sum_{v \in V(G)} |CC_v(a,S,i)| =\frac{1}{|A|} \sum_{a \in A} \Delta(a,S,i) \geq \Delta(S,i)\end{equation*}

With equality iff $A \subseteq \left\{ a \in \mathcal{U}[G] \mid \Delta(a,S,i) = \Delta(S,i)\right\}$.

Therefore, the social welfare of $G=G(S,i)$, computed adding up the utilities $U_v(\bold{s})$ for each agent $v \in V(G)$, given that the adversary attacks the subset of nodes $A$ satisfies the next inequality: 
\begin{equation*} -\alpha |E(G)|-\beta |\mathcal{I}(G)| +\frac{1}{|A|} \sum_{v\in V(G)}\sum_{a\in A}|CC_v(a,S,i)| \geq  -\alpha |E(G)|-\beta |\mathcal{I}(G)| + \Delta(S,i)\end{equation*}

This allows us to deduce that the subset of nodes attacked by the maximum disruption adversary is exactly the subset of vulnerable nodes $a$ from $G$ such that $\Delta(a,S,i) = \Delta(S,i)$ which is what we wanted to see.
\end{proof}
\subsection{The meta-graph $G'$}

One of the first natural properties of a best response of a player $u$ is that $u$ buys at most one link to a specific region of contiguous vulnerable or immunised nodes. Let us introduce the concept of \emph{meta-graph}, that helps us to prove that this is indeed true.

\begin{definition}
Given a graph $G=G(S,i)$ the corresponding  \emph{ meta-graph} $G'=G'(S,i)$ is an undirected graph defined as follows:

(i) The set of vertices of $G'$, called \emph{meta-nodes} are exactly the maximally connected components of $\mathcal{U}[G]$ and   $\mathcal{I}[G]$. More precisely, given a node $w\in V(G)$ the corresponding meta-node containing $w$ in $G=G(S,i)$ is denoted by $W(w,S,i)$.
Then, the set of immunised and vulnerable meta-nodes are denoted by $\mathcal{I}[G']$ and $\mathcal{U}[G']$, respectively.

(ii) The set of edges of $G'$ are exactly the pair of meta-nodes $V_1V_2$ such that there exist nodes $v_1,v_2 \in V(G)$ such that satisfy $v_1 \in V_1$, $v_2 \in V_2$ and $v_1v_2 \in E(G)$. 
\end{definition}

Once the meta-graph has been introduced we consider the meta-node containing $u$, $W(u,S,i) \in V(G')$ which, depending on whether $u$ buys immunisation, is either a meta-node from $\mathcal{I}[G']$ or from $\mathcal{U}[G']$.  Now we take a look to a property that allows to compare the utility of distinct strategies given that there are some similitudes in the subsets $CC_u(v,\cdot,i)$ and the values $\Delta(v,\cdot,i)$ with $v\in \mathcal{U}[G]$.

\begin{lemma}
\label{lemm:trivial}
Let $(S,i), (S',i)$ be two strategies for $u$. If for every $v \in \mathcal{U}[G]$, $CC_u(v,S,i) = CC_u(v,S',i)$ and $\Delta(v,S,i) = \Delta(v,S',i)$,  then $U(u,S,i) - U(u,S',i) = -\alpha \left( |S| - |S'|\right)$. 
\end{lemma}

\begin{proof} If $\Delta(v,S,i) = \Delta(v,S',i)$ for every $v \in \mathcal{U}[G]$ then $A(S,i) = A(S',i)$. If, moreover, $CC_u(v,S,i) = CC_u(v,S',i)$ for every $v\in \mathcal{U}[G]$ then, in particular, $CC_u(v,S,i) = CC_u(v,S',i)$ for every $v \in A(S,i) = A(S',i)$ and from here: 
\begin{equation*}\frac{1}{|A(S,i)|}\sum_{a \in A(S,i)}|CC_u(a,S,i)| = \frac{1}{|A(S',i)|}\sum_{a \in A(S',i)}|CC_u(a,S',i)|\end{equation*}
Then, the conclusion is clear using the definition of the utility $U(u,\cdot,i)$. \end{proof}

\begin{corollary}
\label{corol:simplification1}
Any potential Best Response contains at most one link to the same meta-node.
\end{corollary}

\begin{proof} Let $v_1,v_2 \in V(G)$ be two nodes belonging to the same meta-node and suppose that $(S,i)$ is any strategy for $u$ with  $v_1,v_2 \in S$.  Consider $S' = S \setminus \left\{ v_1\right\}$. Clearly, $CC_u(v,S,i) = CC_u(v,S',i)$ and $\Delta(v,S,i) = \Delta(v,S',i)$ for every $v \in \mathcal{U}[G]$. Then, by Lemma \ref{lemm:trivial}, $U(u,S,i) - U(u,S',i) = -\alpha < 0$. Therefore, $U(u,S',i) > U(u,S,i)$ and $(S,i)$ cannot have maximum utility.
\end{proof}

In view of this we now redefine the following concepts:

(a) Thanks to Corollary \ref{corol:simplification1}, any link-strategy $S$ considered from any potential best response  for $u$ can be assumed without loss of generality to be a subset of $V(G')$, meaning that $u$ is buying one link for each meta-node from this subset and each such link is pointing to any node from $G$ inside such meta-node (it does not matter to which one). 



(b) Notice that any two nodes $z_1,z_2$ belonging to the same meta-node $z \in V(G')$ satisfy $\Delta(z_1,S,i) = \Delta(z_2,S,i)$ for any link-strategy $S$ and immunisation value $i$. Therefore, it makes sense to define $\Delta(z,S,i)$ as the sum of the squares of the size of the connected components we obtain after the removal of all the nodes from $z$ assuming that $u$ has adopted the link strategy $S$ and immunisation value $i$. 

(c) The attacked set $A(S,i)$ can be redefined as a subset of meta-nodes from $G'$ instead of nodes from $V(G)$.

\subsection{The Meta-Tree $T$}

Let $(S,i)$ be a strategy for player $u$ and $G'=G'(S,i)$ the corresponding meta-graph. 

We define $\mathcal{U}_1(G')$ to be the set of vulnerable meta-nodes from $V(G')$ that are not an articulation point from $V(G')$, that is, the set of vulnerable meta-nodes that when removed we obtain exactly the same number of connected components, and $\mathcal{U}_{\geq 2}(G')$ the set of vulnerable meta-nodes from $V(G')$ that are an articulation point in $V(G')$, that is, the set of vulnerable meta-nodes that when removed we obtain at least two distinct connected components in $V(G')$. Then, a connected sub-graph of at least two meta-nodes that remains connected whenever we remove any vulnerable meta-node is called to be a \emph{$2_{\mathcal{I}}-$vertex-connected component} of $G'$.


The next structure is crucial to simplify the problem of computing a Best Response in this model. The concept follows the same idea as in \cite{Lenzner} but we think our definition is much more simple and clear. 

\begin{definition}
The \emph{meta-tree} of $G'$, noted as $T=T(G')$ is  defined in the following way:

(i) The vertices of $T$, called \emph{meta-tree nodes} decompose into two subsets $V_{1}(G')$, $V_{\geq 2}(G')$ that consist of the maximal $2_{\mathcal{I}}$-vertex-connected components from $G'$ and  the meta-nodes from $\mathcal{U}_{\geq 2}(G')$, respectively.

(ii) The edges of $T$, called \emph{meta-tree edges} consist of the pairs $(w,W)$ with $W \in V_1(G')$ and $w \in V_{\geq 2}(G')$ such that $w \in W$. 

\end{definition}

Let $v \in V(T(G'))$. We will always assume that we root $T$ on the node containing $u$. When the context is clear we define $v_1,...,v_{k(v)}$ the children of $v$ with respect $T$ and $T_1(v),...,T_{k(v)}(v)$ the subtrees hanging from $v$. Then we also define $T(v) = \left\{ v \right\} \cup_{1 \leq j \leq k(v)}T_j(v)$ and $\overline{T(v)}$ the tree we obtain when removing $T(v)$ from $T$ which is the same as the connected component in which $W(u,S,i)$ belongs after removing $v$ from the graph. This notation will be useful specially in the next sections. 

\vskip 5pt

In the forthcoming subsections the most common situations that we will be dealing are: either working with the empty link-strategy for player $u$, working with a general strategy $(S,i)$ for player $u$, or comparing between two strategies $(S,i)$, $(S',i')$. In this last case, the most common scenario will consist in comparing $(S,i)$ with $(S',i)$, i.e., two strategies with the same immunisation value. In all these situations it is really important to distinguish between the next three levels of abstraction:

1. The first level of abstraction corresponds to the original network after $u$ adopts $(S,i)$. This network is $G(S,i)$ and the  nodes  from this graph, $V(G(S,i))$, constitute the most basic kind of nodes. If the context is clear we might write $G$ instead of $G(S,i)$.

2. The second level of abstraction corresponds to the meta-graph $G'(S,i)$. This is the network in which we merge connected nodes that are neighbours having the same immunisation value into bigger nodes which we call \emph{meta-nodes}. If the context is clear we might write $G'$ instead of $G'(S,i)$.

3. Finally the third level of abstraction is the meta-tree that is obtained considering the maximal $2_{\mathcal{I}}-$vertex-connected components of the meta-graph. In most of the cases this network will be noted as $T(G'(S,i))$ and in this third level we talk about meta-tree nodes and meta-tree edges. If the context is clear we might write $T$ instead of $T(G')$.

\subsection{Simplifying the Set of Possible Best Responses}

We continue obtaining some results that help us simplify  how best responses  can be assumed to be like. 

\textbf{The delete-simplification.} Scenario (a): Let $(S,i)$ be a strategy for player $u$. Suppose that $v_1,v_2 \in S$ are meta-nodes belonging to the same $2_{\mathcal{I}}$-vertex-connected component from $G'(\emptyset,i)$. Furthermore, suppose that $v_1 \in \mathcal{I}[G'(\emptyset,i)]$ and let $S' = S \setminus \left\{v_1\right\}$.

\begin{center}

\tikzset{every picture/.style={line width=0.75pt}} 

\begin{tikzpicture}[x=0.75pt,y=0.75pt,yscale=-0.75,xscale=0.75]

\draw   (100,167) .. controls (100,141.04) and (121.04,120) .. (147,120) .. controls (172.96,120) and (194,141.04) .. (194,167) .. controls (194,192.96) and (172.96,214) .. (147,214) .. controls (121.04,214) and (100,192.96) .. (100,167) -- cycle ;
\draw    (219,55) .. controls (184.52,62.88) and (133.56,91.13) .. (130.12,153.14) ;
\draw [shift={(130,156)}, rotate = 271.79] [fill={rgb, 255:red, 0; green, 0; blue, 0 }  ][line width=0.08]  [draw opacity=0] (8.93,-4.29) -- (0,0) -- (8.93,4.29) -- cycle    ;
\draw    (219,55) .. controls (195.48,69.7) and (154.67,116.09) .. (156.81,162.18) ;
\draw [shift={(157,165)}, rotate = 265.14] [fill={rgb, 255:red, 0; green, 0; blue, 0 }  ][line width=0.08]  [draw opacity=0] (8.93,-4.29) -- (0,0) -- (8.93,4.29) -- cycle    ;
\draw  [fill={rgb, 255:red, 250; green, 242; blue, 242 }  ,fill opacity=1 ] (213,55) .. controls (213,51.69) and (215.69,49) .. (219,49) .. controls (222.31,49) and (225,51.69) .. (225,55) .. controls (225,58.31) and (222.31,61) .. (219,61) .. controls (215.69,61) and (213,58.31) .. (213,55) -- cycle ;
\draw  [fill={rgb, 255:red, 250; green, 242; blue, 242 }  ,fill opacity=1 ] (153.81,168.19) .. controls (153.81,166.43) and (155.24,165) .. (157,165) .. controls (158.76,165) and (160.19,166.43) .. (160.19,168.19) .. controls (160.19,169.96) and (158.76,171.39) .. (157,171.39) .. controls (155.24,171.39) and (153.81,169.96) .. (153.81,168.19) -- cycle ;
\draw  [fill={rgb, 255:red, 7; green, 0; blue, 0 }  ,fill opacity=1 ] (126.91,159.09) .. controls (126.91,157.38) and (128.29,156) .. (130,156) .. controls (131.71,156) and (133.09,157.38) .. (133.09,159.09) .. controls (133.09,160.8) and (131.71,162.19) .. (130,162.19) .. controls (128.29,162.19) and (126.91,160.8) .. (126.91,159.09) -- cycle ;
\draw    (157,80) -- (163,95) ;
\draw    (152,89) -- (167,85) ;

\draw (115.2,160.6) node [anchor=north west][inner sep=0.75pt]   [align=left] {$\displaystyle  \begin{array}{{>{\displaystyle}l}}
v_{1}\\
\end{array}$};
\draw (135,170) node [anchor=north west][inner sep=0.75pt]   [align=left] {$\displaystyle  \begin{array}{{>{\displaystyle}l}}
v_{2}\\
\end{array}$};
\draw (225,38) node [anchor=north west][inner sep=0.75pt]   [align=left] {$\displaystyle  \begin{array}{{>{\displaystyle}l}}
u\\
\end{array}$};
\draw (169,237) node [anchor=north west][inner sep=0.75pt]   [align=left] {$\displaystyle ( a)$\\};

\end{tikzpicture}

\end{center}

We now examine some properties of this scenario that allow us to simplify the link-strategies to be considered in the forthcoming sections.  

\begin{lemma}
\label{lemm:delete}
Let $v \in \mathcal{U}(G(S,i))$. Then, for any $z \in V(G)$, $CC_z(v,S,i) = CC_z(v,S',i)$ implying  $\Delta(v,S,i) = \Delta(v,S',i)$, too. 
\end{lemma}

\begin{proof} Let us define $S_0 = S \setminus \left( \left\{ v_1 \right\} \cup \left\{v_2\right\} \right)$ so that $S = S_0 \cup \left\{ v_1 \right\} \cup \left\{v_2 \right\}$ and $S' = S_0 \cup \left\{v_2\right\}$. 
Consider the graph $G(S_0,i)$. Let  $CC_v = W(v,S_0,i)$ and let $CC$ be $\left\{CC_v \right\}$ together with the collection of connected components obtained from $G(S_0,i)$ after disconnecting $W(v,S_0,i)$. Then, let $CC_u,CC_{1},CC_2$  the connected components from $CC$ in which $u,v_1,v_2$ belong, respectively. Since $v_1,v_2$ belong to the same maximal $2_{\mathcal{I}}-$vertex-connected component from $G'(\emptyset,i)$ it can only happen that $CC_1=CC_2$. Moreover, $v_1\in \mathcal{I}(G(S,i))$  implies that $CC_1 \neq CC_v$. Then we distinguish only two possible cases:

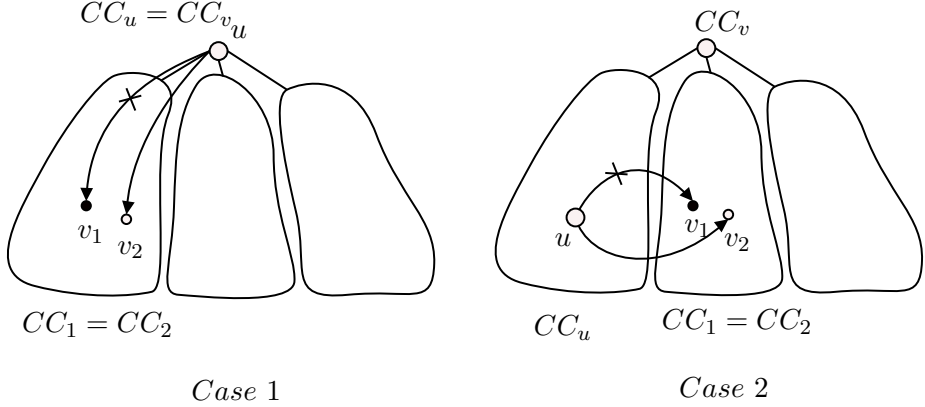
\begin{center}

\tikzset{every picture/.style={line width=0.75pt}} 

\begin{tikzpicture}[x=0.75pt,y=0.75pt,yscale=-0.75,xscale=0.75]

\draw    (180,49) .. controls (147.49,61.81) and (100.44,85.28) .. (97.12,147.15) ;
\draw [shift={(97,150)}, rotate = 271.79] [fill={rgb, 255:red, 0; green, 0; blue, 0 }  ][line width=0.08]  [draw opacity=0] (8.93,-4.29) -- (0,0) -- (8.93,4.29) -- cycle ;
\draw    (180,49) .. controls (161.38,67.62) and (127.39,108.33) .. (124.16,156.07) ;
\draw [shift={(124,159)}, rotate = 278.65] [fill={rgb, 255:red, 0; green, 0; blue, 0 }  ][line width=0.08]  [draw opacity=0] (8.93,-4.29) -- (0,0) -- (8.93,4.29) -- cycle   ;
\draw  [fill={rgb, 255:red, 250; green, 242; blue, 242 }  ,fill opacity=1 ] (180,49) .. controls (180,45.69) and (182.69,43) .. (186,43) .. controls (189.31,43) and (192,45.69) .. (192,49) .. controls (192,52.31) and (189.31,55) .. (186,55) .. controls (182.69,55) and (180,52.31) .. (180,49) -- cycle ;
\draw  [fill={rgb, 255:red, 250; green, 242; blue, 242 }  ,fill opacity=1 ] (120.81,162.19) .. controls (120.81,160.43) and (122.24,159) .. (124,159) .. controls (125.76,159) and (127.19,160.43) .. (127.19,162.19) .. controls (127.19,163.96) and (125.76,165.39) .. (124,165.39) .. controls (122.24,165.39) and (120.81,163.96) .. (120.81,162.19) -- cycle ;
\draw  [fill={rgb, 255:red, 7; green, 0; blue, 0 }  ,fill opacity=1 ] (93.91,153.09) .. controls (93.91,151.38) and (95.29,150) .. (97,150) .. controls (98.71,150) and (100.09,151.38) .. (100.09,153.09) .. controls (100.09,154.8) and (98.71,156.19) .. (97,156.19) .. controls (95.29,156.19) and (93.91,154.8) .. (93.91,153.09) -- cycle ;
\draw    (124,74) -- (130,89) ;
\draw    (119,83) -- (134,79) ;
\draw   (115,66) .. controls (125,61) and (164,62) .. (152,92) .. controls (140,122) and (154,208) .. (141,211) .. controls (128,214) and (97,216) .. (62,207) .. controls (27,198) and (52,156.25) .. (68,127) .. controls (84,97.75) and (105,71) .. (115,66) -- cycle ;
\draw   (232,75) .. controls (242,70) and (276,62) .. (292,105) .. controls (308,148) and (332,162) .. (331,192) .. controls (330,222) and (316,209) .. (273,213) .. controls (230,217) and (247,195) .. (238,157) .. controls (229,119) and (222,80) .. (232,75) -- cycle ;
\draw   (171,72) .. controls (184,51) and (223,79) .. (223,119) .. controls (223,159) and (238,187) .. (237,204) .. controls (236,221) and (173,214) .. (159,213) .. controls (145,212) and (154.75,191.13) .. (156,154) .. controls (157.25,116.88) and (158,93) .. (171,72) -- cycle ;
\draw    (186,55) -- (189,66) ;
\draw    (192,49) -- (232,75) ;
\draw    (180,49) -- (147,69) ;
\draw    (426,161) .. controls (435.85,137.36) and (468.02,109.84) .. (503.38,148.2) ;
\draw [shift={(505,150)}, rotate = 228.72] [fill={rgb, 255:red, 0; green, 0; blue, 0 }  ][line width=0.08]  [draw opacity=0] (8.93,-4.29) -- (0,0) -- (8.93,4.29) -- cycle    ;
\draw    (426,161) .. controls (433.88,190.55) and (483.48,203.61) .. (526.83,164.04) ;
\draw [shift={(528.81,162.19)}, rotate = 143.29] [fill={rgb, 255:red, 0; green, 0; blue, 0 }  ][line width=0.08]  [draw opacity=0] (8.93,-4.29) -- (0,0) -- (8.93,4.29) -- cycle    ;
\draw  [fill={rgb, 255:red, 250; green, 242; blue, 242 }  ,fill opacity=1 ] (508,47) .. controls (508,43.69) and (510.69,41) .. (514,41) .. controls (517.31,41) and (520,43.69) .. (520,47) .. controls (520,50.31) and (517.31,53) .. (514,53) .. controls (510.69,53) and (508,50.31) .. (508,47) -- cycle ;
\draw  [fill={rgb, 255:red, 250; green, 242; blue, 242 }  ,fill opacity=1 ] (525.61,159) .. controls (525.61,157.24) and (527.04,155.81) .. (528.81,155.81) .. controls (530.57,155.81) and (532,157.24) .. (532,159) .. controls (532,160.76) and (530.57,162.19) .. (528.81,162.19) .. controls (527.04,162.19) and (525.61,160.76) .. (525.61,159) -- cycle ;
\draw  [fill={rgb, 255:red, 7; green, 0; blue, 0 }  ,fill opacity=1 ] (501.91,153.09) .. controls (501.91,151.38) and (503.29,150) .. (505,150) .. controls (506.71,150) and (508.09,151.38) .. (508.09,153.09) .. controls (508.09,154.8) and (506.71,156.19) .. (505,156.19) .. controls (503.29,156.19) and (501.91,154.8) .. (501.91,153.09) -- cycle ;
\draw    (451,124) -- (457,139) ;
\draw    (446,131) -- (461,132) ;
\draw   (443,64) .. controls (453,59) and (492,60) .. (480,90) .. controls (468,120) and (482,206) .. (469,209) .. controls (456,212) and (425,214) .. (390,205) .. controls (355,196) and (380,154.25) .. (396,125) .. controls (412,95.75) and (433,69) .. (443,64) -- cycle ;
\draw   (560,73) .. controls (570,68) and (604,60) .. (620,103) .. controls (636,146) and (660,160) .. (659,190) .. controls (658,220) and (644,207) .. (601,211) .. controls (558,215) and (575,193) .. (566,155) .. controls (557,117) and (550,78) .. (560,73) -- cycle ;
\draw   (499,70) .. controls (512,49) and (551,77) .. (551,117) .. controls (551,157) and (566,185) .. (565,202) .. controls (564,219) and (501,212) .. (487,211) .. controls (473,210) and (482.75,189.13) .. (484,152) .. controls (485.25,114.88) and (486,91) .. (499,70) -- cycle ;
\draw    (514,53) -- (517,64) ;
\draw    (520,47) -- (560,73) ;
\draw    (508,47) -- (475,67) ;
\draw  [fill={rgb, 255:red, 250; green, 242; blue, 242 }  ,fill opacity=1 ] (420,161) .. controls (420,157.69) and (422.69,155) .. (426,155) .. controls (429.31,155) and (432,157.69) .. (432,161) .. controls (432,164.31) and (429.31,167) .. (426,167) .. controls (422.69,167) and (420,164.31) .. (420,161) -- cycle ;

\draw (89.2,164.6) node [anchor=north west][inner sep=0.75pt]   [align=left] {$v_1$};
\draw (116,175) node [anchor=north west][inner sep=0.75pt]   [align=left] {$v_2$};
\draw (192,30) node [anchor=north west][inner sep=0.75pt]   [align=left] {$u$};
\draw (52,223) node [anchor=north west][inner sep=0.75pt]   [align=left] {$CC_1=CC_2$};
\draw (91,14) node [anchor=north west][inner sep=0.75pt]   [align=left] {$CC_{u} =CC_{v}$};
\draw (157,265) node [anchor=north west][inner sep=0.75pt]   [align=left] {$\displaystyle  \begin{array}{{>{\displaystyle}l}}
Case\ 1\\
\end{array}$};
\draw (497.2,160.6) node [anchor=north west][inner sep=0.75pt]   [align=left] {$v_1$};
\draw (524,170) node [anchor=north west][inner sep=0.75pt]   [align=left] {$v_2$};
\draw (410,170) node [anchor=north west][inner sep=0.75pt]   [align=left] {$u$};
\draw (482,221) node [anchor=north west][inner sep=0.75pt]   [align=left] {$CC_1=CC_2$};
\draw (504,20) node [anchor=north west][inner sep=0.75pt]   [align=left] {$CC_v$};
\draw (485,263) node [anchor=north west][inner sep=0.75pt]   [align=left] {$\displaystyle  \begin{array}{{>{\displaystyle}l}}
Case\ 2\\
\end{array}$};
\draw (396,225) node [anchor=north west][inner sep=0.75pt]   [align=left] {$CC_u$};

\end{tikzpicture}

\end{center}

(1) $CC_u = CC_v$. Then it is clear from the figure that for any $z \in V(G)$, $CC_z(v,S,i) = CC_z(v,S',i)$ because  $v_1$ is immunised by hypothesis, implying $\Delta(v,S',i) =\Delta(v,S,i)$. 

(2) $CC_u \neq CC_v$. Then, again, it is clear from the figure that for any $z \in V(G)$, $CC_z(v,S,i) = CC_z(v,S',i)$, implying $\Delta(v,S',i) =\Delta(v,S,i)$, too. 
\end{proof}

\textbf{The swap-simplification.}  We now examine some other scenarios  in which we cap apply a swap movement. Let $(S,i)$ be a strategy for player $u$. Then we consider:

Scenario (b.i): Suppose that $v_1,v_2 \in \mathcal{I}[G'(\emptyset,i)]$. Suppose that $v_1 \in S$, $v_2 \not \in S$ and consider the strategy $S' = S \cup \left\{v_2\right\} \setminus \left\{v_1\right\}$. Furthermore, suppose that $v_1,v_2$ are meta-nodes belonging to the same maximal  $2_{\mathcal{I}}$-vertex-connected component from $G'(\emptyset,i)$.

Scenario (b.ii): Suppose that $v_1 \in \mathcal{U}_1[G'(\emptyset,i)]$ and $v_2 \in \mathcal{I}[G'(\emptyset,i)]$. Suppose that $v_1 \in S$, $v_2 \not \in S$ and consider the link-strategy $S' = S \cup \left\{v_2\right\} \setminus \left\{v_1\right\}$.   Furthermore, suppose that $v_1,v_2$ are meta-nodes belonging to the same maximal $2_{\mathcal{I}}$-vertex-connected component from $G'(\emptyset,i)$.

Scenario (b.iii): Suppose  that $v_1 \in \mathcal{U}_{\geq 2}[G'(\emptyset,i)]$, $v_2 \in \mathcal{I}[G'(\emptyset,i)]$ and $v_2$ is any neighbour of $v_1$ in $G'(\emptyset,i)$ contained in any simple and connected path (in $G'(\emptyset,i)$ which we assume to be connected) from $v_1$ to $u$.  Consider the strategy $S' = S \cup \left\{v_2\right\} \setminus \left\{v_1\right\}$.   Furthermore, suppose that $v_1,v_2$ are meta-nodes belonging to the same maximal $2_{\mathcal{I}}$-vertex-connected component from $G'(\emptyset,i)$.

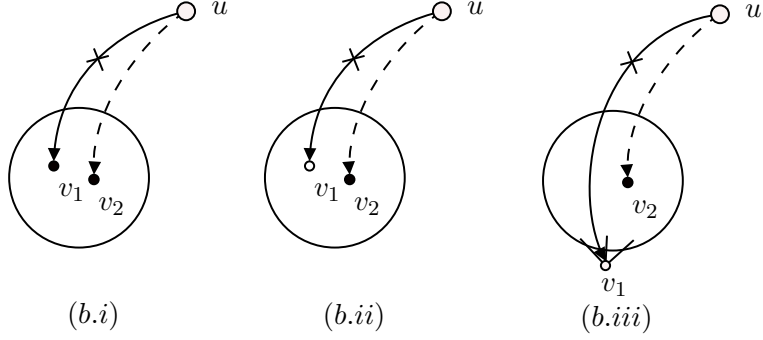
\begin{center}

\tikzset{every picture/.style={line width=0.75pt}} 

\begin{tikzpicture}[x=0.75pt,y=0.75pt,yscale=-0.75,xscale=0.75]

\draw   (99,158) .. controls (99,132.04) and (120.04,111) .. (146,111) .. controls (171.96,111) and (193,132.04) .. (193,158) .. controls (193,183.96) and (171.96,205) .. (146,205) .. controls (120.04,205) and (99,183.96) .. (99,158) -- cycle ;
\draw    (218,46) .. controls (183.52,53.88) and (132.56,82.13) .. (129.12,144.14) ;
\draw [shift={(129,147)}, rotate = 271.79] [fill={rgb, 255:red, 0; green, 0; blue, 0 }  ][line width=0.08]  [draw opacity=0] (8.93,-4.29) -- (0,0) -- (8.93,4.29) -- cycle    ;
\draw  [dash pattern={on 4.5pt off 4.5pt}]  (218,46) .. controls (194.48,60.7) and (153.67,107.09) .. (155.81,153.18) ;
\draw [shift={(156,156)}, rotate = 273.6] [fill={rgb, 255:red, 0; green, 0; blue, 0 }  ][line width=0.08]  [draw opacity=0] (8.93,-4.29) -- (0,0) -- (8.93,4.29) -- cycle    ;
\draw  [fill={rgb, 255:red, 250; green, 242; blue, 242 }  ,fill opacity=1 ] (212,46) .. controls (212,42.69) and (214.69,40) .. (218,40) .. controls (221.31,40) and (224,42.69) .. (224,46) .. controls (224,49.31) and (221.31,52) .. (218,52) .. controls (214.69,52) and (212,49.31) .. (212,46) -- cycle ;
\draw  [fill={rgb, 255:red, 7; green, 0; blue, 0 }  ,fill opacity=1 ] (152.81,159.19) .. controls (152.81,157.43) and (154.24,156) .. (156,156) .. controls (157.76,156) and (159.19,157.43) .. (159.19,159.19) .. controls (159.19,160.96) and (157.76,162.39) .. (156,162.39) .. controls (154.24,162.39) and (152.81,160.96) .. (152.81,159.19) -- cycle ;
\draw  [fill={rgb, 255:red, 7; green, 0; blue, 0 }  ,fill opacity=1 ] (125.91,150.09) .. controls (125.91,148.38) and (127.29,147) .. (129,147) .. controls (130.71,147) and (132.09,148.38) .. (132.09,150.09) .. controls (132.09,151.8) and (130.71,153.19) .. (129,153.19) .. controls (127.29,153.19) and (125.91,151.8) .. (125.91,150.09) -- cycle ;
\draw    (156,71) -- (162,86) ;
\draw    (151,80) -- (166,76) ;
\draw   (271,158) .. controls (271,132.04) and (292.04,111) .. (318,111) .. controls (343.96,111) and (365,132.04) .. (365,158) .. controls (365,183.96) and (343.96,205) .. (318,205) .. controls (292.04,205) and (271,183.96) .. (271,158) -- cycle ;
\draw    (390,46) .. controls (355.53,53.88) and (304.56,82.13) .. (301.12,144.14) ;
\draw [shift={(301,147)}, rotate = 271.79] [fill={rgb, 255:red, 0; green, 0; blue, 0 }  ][line width=0.08]  [draw opacity=0] (8.93,-4.29) -- (0,0) -- (8.93,4.29) -- cycle    ;
\draw  [dash pattern={on 4.5pt off 4.5pt}]  (390,46) .. controls (366.48,60.7) and (325.67,107.09) .. (327.81,153.18) ;
\draw [shift={(328,156)}, rotate = 265.14] [fill={rgb, 255:red, 0; green, 0; blue, 0 }  ][line width=0.08]  [draw opacity=0] (8.93,-4.29) -- (0,0) -- (8.93,4.29) -- cycle    ;
\draw  [fill={rgb, 255:red, 250; green, 242; blue, 242 }  ,fill opacity=1 ] (384,46) .. controls (384,42.69) and (386.69,40) .. (390,40) .. controls (393.31,40) and (396,42.69) .. (396,46) .. controls (396,49.31) and (393.31,52) .. (390,52) .. controls (386.69,52) and (384,49.31) .. (384,46) -- cycle ;
\draw  [fill={rgb, 255:red, 7; green, 0; blue, 0 }  ,fill opacity=1 ] (324.81,159.19) .. controls (324.81,157.43) and (326.24,156) .. (328,156) .. controls (329.76,156) and (331.19,157.43) .. (331.19,159.19) .. controls (331.19,160.96) and (329.76,162.39) .. (328,162.39) .. controls (326.24,162.39) and (324.81,160.96) .. (324.81,159.19) -- cycle ;
\draw  [fill={rgb, 255:red, 255; green, 255; blue, 255 }  ,fill opacity=1 ] (297.91,150.09) .. controls (297.91,148.38) and (299.29,147) .. (301,147) .. controls (302.71,147) and (304.09,148.38) .. (304.09,150.09) .. controls (304.09,151.8) and (302.71,153.19) .. (301,153.19) .. controls (299.29,153.19) and (297.91,151.8) .. (297.91,150.09) -- cycle ;
\draw    (328,71) -- (334,86) ;
\draw    (323,80) -- (338,76) ;
\draw   (458,160) .. controls (458,134.04) and (479.04,113) .. (505,113) .. controls (530.96,113) and (552,134.04) .. (552,160) .. controls (552,185.96) and (530.96,207) .. (505,207) .. controls (479.04,207) and (458,185.96) .. (458,160) -- cycle ;
\draw    (577,48) .. controls (495.24,57.85) and (474.61,162.78) .. (498.86,211.8) ;
\draw [shift={(500,214)}, rotate = 241.56] [fill={rgb, 255:red, 0; green, 0; blue, 0 }  ][line width=0.08]  [draw opacity=0] (8.93,-4.29) -- (0,0) -- (8.93,4.29) -- cycle    ;
\draw   [dash pattern={on 4.5pt off 4.5pt}] (577,48) .. controls (553.48,62.7) and (512.67,109.09) .. (514.81,155.18) ;
\draw [shift={(515,158)}, rotate = 265.14] [fill={rgb, 255:red, 0; green, 0; blue, 0 }  ][line width=0.08]  [draw opacity=0] (8.93,-4.29) -- (0,0) -- (8.93,4.29) -- cycle    ;
\draw  [fill={rgb, 255:red, 250; green, 242; blue, 242 }  ,fill opacity=1 ] (571,48) .. controls (571,44.69) and (573.69,42) .. (577,42) .. controls (580.31,42) and (583,44.69) .. (583,48) .. controls (583,51.31) and (580.31,54) .. (577,54) .. controls (573.69,54) and (571,51.31) .. (571,48) -- cycle ;
\draw  [fill={rgb, 255:red, 7; green, 0; blue, 0 } ,fill opacity=1 ] (511.81,161.19) .. controls (511.81,159.43) and (513.24,158) .. (515,158) .. controls (516.76,158) and (518.19,159.43) .. (518.19,161.19) .. controls (518.19,162.96) and (516.76,164.39) .. (515,164.39) .. controls (513.24,164.39) and (511.81,162.96) .. (511.81,161.19) -- cycle ;
\draw    (515,73) -- (521,88) ;
\draw    (510,82) -- (525,78) ;
\draw    (500,217.09) -- (483,199) ;
\draw    (501,197) -- (500,217.09) ;
\draw    (519,200) -- (500,217.09) ;
\draw  [fill={rgb, 255:red, 250; green, 242; blue, 242 }  ,fill opacity=1 ] (496.91,217.09) .. controls (496.91,215.38) and (498.29,214) .. (500,214) .. controls (501.71,214) and (503.09,215.38) .. (503.09,217.09) .. controls (503.09,218.8) and (501.71,220.19) .. (500,220.19) .. controls (498.29,220.19) and (496.91,218.8) .. (496.91,217.09) -- cycle ;

\draw (121.2,151.6) node [anchor=north west][inner sep=0.75pt]   [align=left] {$\displaystyle  \begin{array}{{>{\displaystyle}l}}
v_{1}\\
\end{array}$};
\draw (148,161) node [anchor=north west][inner sep=0.75pt]   [align=left] {$\displaystyle  \begin{array}{{>{\displaystyle}l}}
v_{2}\\
\end{array}$};
\draw (224,29) node [anchor=north west][inner sep=0.75pt]   [align=left] {$\displaystyle  \begin{array}{{>{\displaystyle}l}}
u\\
\end{array}$};
\draw (136,239) node [anchor=north west][inner sep=0.75pt]   [align=left] {$\displaystyle ( b.i)$\\};
\draw (293.2,151.6) node [anchor=north west][inner sep=0.75pt]   [align=left] {$\displaystyle  \begin{array}{{>{\displaystyle}l}}
v_{1}\\
\end{array}$};
\draw (320,161) node [anchor=north west][inner sep=0.75pt]   [align=left] {$\displaystyle  \begin{array}{{>{\displaystyle}l}}
v_{2}\\
\end{array}$};
\draw (396,29) node [anchor=north west][inner sep=0.75pt]   [align=left] {$\displaystyle  \begin{array}{{>{\displaystyle}l}}
u\\
\end{array}$};
\draw (308,240) node [anchor=north west][inner sep=0.75pt]   [align=left] {$\displaystyle ( b.ii)$\\};
\draw (486.2,215.6) node [anchor=north west][inner sep=0.75pt]   [align=left] {$\displaystyle  \begin{array}{{>{\displaystyle}l}}
v_{1}\\
\end{array}$};
\draw (507,163) node [anchor=north west][inner sep=0.75pt]   [align=left] {$\displaystyle  \begin{array}{{>{\displaystyle}l}}
v_{2}\\
\end{array}$};
\draw (583,31) node [anchor=north west][inner sep=0.75pt]   [align=left] {$\displaystyle  \begin{array}{{>{\displaystyle}l}}
u\\
\end{array}$};
\draw (481,241) node [anchor=north west][inner sep=0.75pt]   [align=left] {$\displaystyle ( b.iii)$\\};

\end{tikzpicture}

\end{center}

\begin{lemma}
\label{lemm:swap}
Let $v \in \mathcal{U}(G(S,i))$. Then, for any $z \in V(G)$, $CC_z(v,S,i) = CC_z(v,S',i)$ if $v \not \in W(u,S,i)$ and $CC_z(v,S,i) \subseteq CC_z(v,S',i)$ if $v \in W(u,S,i)$. This implies that  if $v \not \in W(u,S,i)$, then $\Delta(v,S,i) = \Delta(v,S',i)$. Otherwise, $\Delta(v,S,i) \leq \Delta(v,S',i)$.
\end{lemma}

\begin{proof} Let $S_0 = S \setminus \left\{ v_1\right\}$ so that $S = S_0 \cup \left\{ v_1 \right\}$ and $S' = S_0 \cup \left\{ v_2 \right\}$. 
Consider the graph $G(S_0,i)$. Let  $CC_v = W(v,S_0,i)$ and let $CC$ be $\left\{CC_v \right\}$ together with the collection of connected components obtained from $G(S_0,i)$ after disconnecting $W(v,S_0,i)$. Then, let $CC_u,CC_{1},CC_2$  the connected components from $CC$ in which $u,v_1,v_2$ belong, respectively. First of all, since $v_2\in \mathcal{I}(G(S,i))$ then $CC_2 \neq CC_v$. Then it can only happen the following:

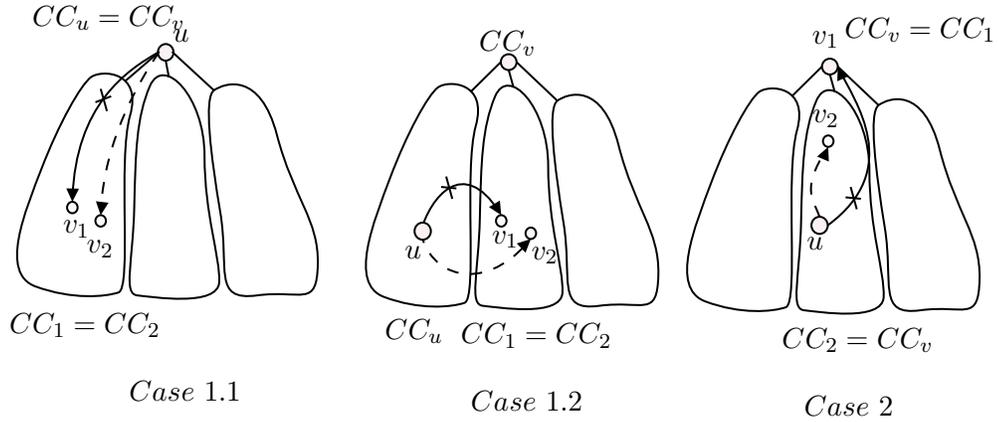
\begin{center}

\tikzset{every picture/.style={line width=0.75pt}} 

\begin{tikzpicture}[x=0.75pt,y=0.75pt,yscale=-0.75,xscale=0.75]

\draw    (101.46,49) .. controls (78.53,61.8) and (45.34,85.28) .. (43,147.15) ;
\draw [shift={(42.91,150)}, rotate = 271.26] [fill={rgb, 255:red, 0; green, 0; blue, 0 }  ][line width=0.08]  [draw opacity=0] (8.93,-4.29) -- (0,0) -- (8.93,4.29) -- cycle    ;
\draw  [dash pattern={on 4.5pt off 4.5pt}]  (101.46,49) .. controls (88.33,67.62) and (64.35,108.33) .. (62.07,156.07) ;
\draw [shift={(61.96,159)}, rotate = 276.63] [fill={rgb, 255:red, 0; green, 0; blue, 0 }  ][line width=0.08]  [draw opacity=0] (8.93,-4.29) -- (0,0) -- (8.93,4.29) -- cycle    ;
\draw  [fill={rgb, 255:red, 250; green, 242; blue, 242 }  ,fill opacity=1 ] (100.43,49.42) .. controls (100.43,46.33) and (102.79,43.83) .. (105.69,43.83) .. controls (108.6,43.83) and (110.96,46.33) .. (110.96,49.42) .. controls (110.96,52.5) and (108.6,55) .. (105.69,55) .. controls (102.79,55) and (100.43,52.5) .. (100.43,49.42) -- cycle ;
\draw  [fill={rgb, 255:red, 253; green, 249; blue, 249 }  ,fill opacity=1 ] (39.37,153.85) .. controls (39.37,151.72) and (40.95,150) .. (42.91,150) .. controls (44.87,150) and (46.46,151.72) .. (46.46,153.85) .. controls (46.46,155.98) and (44.87,157.7) .. (42.91,157.7) .. controls (40.95,157.7) and (39.37,155.98) .. (39.37,153.85) -- cycle ;
\draw    (61.96,74) -- (66.19,89) ;
\draw    (58.43,83) -- (69.01,79) ;
\draw   (55.61,66) .. controls (62.67,61) and (90.18,62) .. (81.71,92) .. controls (73.25,122) and (83.12,208) .. (73.95,211) .. controls (64.78,214) and (42.91,216) .. (18.22,207) .. controls (-6.46,198) and (11.17,156.25) .. (22.46,127) .. controls (33.74,97.75) and (48.56,71) .. (55.61,66) -- cycle ;
\draw   (138.14,75) .. controls (145.2,70) and (169.18,62) .. (180.47,105) .. controls (191.75,148) and (208.68,162) .. (207.98,192) .. controls (207.27,222) and (197.4,209) .. (167.07,213) .. controls (136.73,217) and (148.72,195) .. (142.38,157) .. controls (136.03,119) and (131.09,80) .. (138.14,75) -- cycle ;
\draw   (95.11,72) .. controls (104.28,51) and (131.79,79) .. (131.79,119) .. controls (131.79,159) and (142.38,187) .. (141.67,204) .. controls (140.97,221) and (96.52,214) .. (86.65,213) .. controls (76.77,212) and (83.65,191.12) .. (84.53,154) .. controls (85.41,116.87) and (85.94,93) .. (95.11,72) -- cycle ;
\draw    (105.69,55) -- (107.81,66) ;
\draw    (109.93,49) -- (138.14,75) ;
\draw    (101.46,49) -- (78.18,69) ;
\draw    (277.14,170) .. controls (283.86,146.48) and (305.74,119.12) .. (329.85,156.62) ;
\draw [shift={(331.33,159)}, rotate = 238.94] [fill={rgb, 255:red, 0; green, 0; blue, 0 }  ][line width=0.08]  [draw opacity=0] (8.93,-4.29) -- (0,0) -- (8.93,4.29) -- cycle    ;
\draw  [dash pattern={on 4.5pt off 4.5pt}]  (277.14,170) .. controls (282.52,199.4) and (318.08,212.84) .. (349.31,177.28) ;
\draw [shift={(351.21,175.04)}, rotate = 136.81] [fill={rgb, 255:red, 0; green, 0; blue, 0 }  ][line width=0.08]  [draw opacity=0] (8.93,-4.29) -- (0,0) -- (8.93,4.29) -- cycle    ;
\draw  [fill={rgb, 255:red, 250; green, 242; blue, 242 }  ,fill opacity=1 ] (331.01,56) .. controls (331.01,53.06) and (333.39,50.68) .. (336.32,50.68) .. controls (339.25,50.68) and (341.62,53.06) .. (341.62,56) .. controls (341.62,58.94) and (339.25,61.32) .. (336.32,61.32) .. controls (333.39,61.32) and (331.01,58.94) .. (331.01,56) -- cycle ;
\draw    (294.29,133) -- (298.4,148) ;
\draw    (290.86,140) -- (301.15,141) ;
\draw   (288.8,73) .. controls (295.66,68) and (322.41,69) .. (314.18,99) .. controls (305.95,129) and (315.55,215) .. (306.64,218) .. controls (297.72,221) and (276.45,223) .. (252.44,214) .. controls (228.43,205) and (245.58,163.25) .. (256.56,134) .. controls (267.53,104.75) and (281.94,78) .. (288.8,73) -- cycle ;
\draw   (369.06,82) .. controls (375.92,77) and (399.25,69) .. (410.22,112) .. controls (421.2,155) and (437.67,169) .. (436.98,199) .. controls (436.29,229) and (426.69,216) .. (397.19,220) .. controls (367.69,224) and (379.35,202) .. (373.18,164) .. controls (367.01,126) and (362.2,87) .. (369.06,82) -- cycle ;
\draw   (327.22,79) .. controls (336.14,58) and (362.89,86) .. (362.89,126) .. controls (362.89,166) and (373.18,194) .. (372.49,211) .. controls (371.81,228) and (328.59,221) .. (318.98,220) .. controls (309.38,219) and (316.07,198.13) .. (316.93,161) .. controls (317.78,123.87) and (318.3,100) .. (327.22,79) -- cycle ;
\draw    (336.32,61.32) -- (339.57,73) ;
\draw    (341.62,56) -- (369.06,82) ;
\draw    (331.01,56) -- (310.75,76) ;
\draw  [fill={rgb, 255:red, 250; green, 242; blue, 242 }  ,fill opacity=1 ] (273.02,169.77) .. controls (273.02,166.58) and (275.48,164) .. (278.51,164) .. controls (281.54,164) and (284,166.58) .. (284,169.77) .. controls (284,172.95) and (281.54,175.53) .. (278.51,175.53) .. controls (275.48,175.53) and (273.02,172.95) .. (273.02,169.77) -- cycle ;
\draw  [fill={rgb, 255:red, 253; green, 249; blue, 249 }  ,fill opacity=1 ] (58.41,162.85) .. controls (58.41,160.72) and (60,159) .. (61.96,159) .. controls (63.92,159) and (65.51,160.72) .. (65.51,162.85) .. controls (65.51,164.98) and (63.92,166.7) .. (61.96,166.7) .. controls (60,166.7) and (58.41,164.98) .. (58.41,162.85) -- cycle ;
\draw  [fill={rgb, 255:red, 253; green, 249; blue, 249 }  ,fill opacity=1 ] (327.79,162.85) .. controls (327.79,160.72) and (329.37,159) .. (331.33,159) .. controls (333.29,159) and (334.88,160.72) .. (334.88,162.85) .. controls (334.88,164.98) and (333.29,166.7) .. (331.33,166.7) .. controls (329.37,166.7) and (327.79,164.98) .. (327.79,162.85) -- cycle ;
\draw  [fill={rgb, 255:red, 253; green, 249; blue, 249 }  ,fill opacity=1 ] (347.66,171.19) .. controls (347.66,169.07) and (349.25,167.34) .. (351.21,167.34) .. controls (353.17,167.34) and (354.76,169.07) .. (354.76,171.19) .. controls (354.76,173.32) and (353.17,175.04) .. (351.21,175.04) .. controls (349.25,175.04) and (347.66,173.32) .. (347.66,171.19) -- cycle ;
\draw    (550.88,165.85) .. controls (602.34,127.98) and (568.93,81.92) .. (558.77,61.43) ;
\draw [shift={(557.62,59)}, rotate = 66.28] [fill={rgb, 255:red, 0; green, 0; blue, 0 }  ][line width=0.08]  [draw opacity=0] (8.93,-4.29) -- (0,0) -- (8.93,4.29) -- cycle    ;
\draw  [dash pattern={on 4.5pt off 4.5pt}]  (545.39,160.08) .. controls (541.59,149.15) and (532.47,142) .. (549.83,115.15) ;
\draw [shift={(551.21,113.04)}, rotate = 113.38] [fill={rgb, 255:red, 0; green, 0; blue, 0 }  ][line width=0.08]  [draw opacity=0] (8.93,-4.29) -- (0,0) -- (8.93,4.29) -- cycle    ;
\draw  [fill={rgb, 255:red, 250; green, 242; blue, 242 }  ,fill opacity=1 ] (547.01,59) .. controls (547.01,56.06) and (549.39,53.68) .. (552.32,53.68) .. controls (555.25,53.68) and (557.62,56.06) .. (557.62,59) .. controls (557.62,61.94) and (555.25,64.32) .. (552.32,64.32) .. controls (549.39,64.32) and (547.01,61.94) .. (547.01,59) -- cycle ;
\draw    (566.29,140) -- (570.4,155) ;
\draw    (562.86,147) -- (573.15,148) ;
\draw   (504.8,76) .. controls (511.66,71) and (538.41,72) .. (530.18,102) .. controls (521.95,132) and (531.55,218) .. (522.64,221) .. controls (513.72,224) and (492.45,226) .. (468.44,217) .. controls (444.43,208) and (461.58,166.25) .. (472.56,137) .. controls (483.53,107.75) and (497.94,81) .. (504.8,76) -- cycle ;
\draw   (585.06,85) .. controls (591.92,80) and (615.25,72) .. (626.22,115) .. controls (637.2,158) and (653.67,172) .. (652.98,202) .. controls (652.29,232) and (642.69,219) .. (613.19,223) .. controls (583.69,227) and (595.35,205) .. (589.18,167) .. controls (583.01,129) and (578.2,90) .. (585.06,85) -- cycle ;
\draw   (543.22,82) .. controls (552.14,61) and (578.89,89) .. (578.89,129) .. controls (578.89,169) and (589.18,197) .. (588.49,214) .. controls (587.81,231) and (544.59,224) .. (534.98,223) .. controls (525.38,222) and (532.07,201.13) .. (532.93,164) .. controls (533.78,126.87) and (534.3,103) .. (543.22,82) -- cycle ;
\draw    (552.32,64.32) -- (555.57,76) ;
\draw    (557.62,59) -- (585.06,85) ;
\draw    (547.01,59) -- (526.75,79) ;
\draw  [fill={rgb, 255:red, 250; green, 242; blue, 242 }  ,fill opacity=1 ] (539.9,165.85) .. controls (539.9,162.67) and (542.36,160.08) .. (545.39,160.08) .. controls (548.42,160.08) and (550.88,162.67) .. (550.88,165.85) .. controls (550.88,169.03) and (548.42,171.62) .. (545.39,171.62) .. controls (542.36,171.62) and (539.9,169.03) .. (539.9,165.85) -- cycle ;
\draw  [fill={rgb, 255:red, 253; green, 249; blue, 249 }  ,fill opacity=1 ] (547.66,109.19) .. controls (547.66,107.07) and (549.25,105.34) .. (551.21,105.34) .. controls (553.17,105.34) and (554.76,107.07) .. (554.76,109.19) .. controls (554.76,111.32) and (553.17,113.04) .. (551.21,113.04) .. controls (549.25,113.04) and (547.66,111.32) .. (547.66,109.19) -- cycle ;

\draw (34.76,160.6) node [anchor=north west][inner sep=0.75pt]   [align=left] {$v_1$};
\draw (50.67,174) node [anchor=north west][inner sep=0.75pt]   [align=left] {$v_2$};
\draw (108.31,32) node [anchor=north west][inner sep=0.75pt]   [align=left] {$u$};
\draw (-1.35,223) node [anchor=north west][inner sep=0.75pt]   [align=left] {$CC_1=CC_2$};
\draw (14.01,16) node [anchor=north west][inner sep=0.75pt]   [align=left] {$CC_{u} =CC_{v}$};
\draw (70.4,265) node [anchor=north west][inner sep=0.75pt]   [align=left] {$\displaystyle  \begin{array}{{>{\displaystyle}l}}
Case\ 1.1\\
\\
\end{array}$};
\draw (323.16,167.6) node [anchor=north west][inner sep=0.75pt]   [align=left] {$v_1$};
\draw (349.66,178.19) node [anchor=north west][inner sep=0.75pt]   [align=left] {$v_2$};
\draw (264.43,178) node [anchor=north west][inner sep=0.75pt]   [align=left] {$u$};
\draw (302.21,230) node [anchor=north west][inner sep=0.75pt]   [align=left] {$CC_{1} =CC_{2}$};
\draw (314.47,33) node [anchor=north west][inner sep=0.75pt]   [align=left] {$CC_v$};
\draw (300.19,272) node [anchor=north west][inner sep=0.75pt]   [align=left] {$\displaystyle  \begin{array}{{>{\displaystyle}l}}
Case\ 1.2\\
\\
\end{array}$};
\draw (251.22,227) node [anchor=north west][inner sep=0.75pt]   [align=left] {$CC_u$};
\draw (538.66,86.19) node [anchor=north west][inner sep=0.75pt]   [align=left] {$v_2$};
\draw (534.93,172) node [anchor=north west][inner sep=0.75pt]   [align=left] {$u$};
\draw (518.21,233) node [anchor=north west][inner sep=0.75pt]   [align=left] {$CC_2=CC_v$};
\draw (538.47,33) node [anchor=north west][inner sep=0.75pt]   [align=left] {$v_1$};
\draw (524.19,275) node [anchor=north west][inner sep=0.75pt]   [align=left] {$\displaystyle  \begin{array}{{>{\displaystyle}l}}
Case\ 2\\
\\
\end{array}$};
\draw (560.21,24) node [anchor=north west][inner sep=0.75pt]   [align=left] {$CC_v=CC_1$};

\end{tikzpicture}

\end{center}

(1) $CC_1 = CC_2$. We distinguish two subcases:

(1.1) $CC_u = CC_v$. Then it is clear from the figure that for any $z \in V(G)$, $CC_z(v,S,i) \subseteq CC_z(v,S',i)$. Then it happens $\Delta(v,S,i) \leq \Delta(v,S',i)$. Notice that in this situation it holds $u \in W(v_1,S,i) = W(v,S,i)$ so $v \in W(u,S,i)$. 

(1.2) $CC_u \neq CC_v$. Then looking at the figure above we deduce that for any $z \in V(G)$, $CC_z(v,S,i) = CC_z(v,S',i)$. Then we clearly have $\Delta(v,S,i) = \Delta(v,S',i)$. Notice that in this situation it holds $u \not \in W(v,S,i)$ so $v \not \in W(u,S,i)$.

(2) $CC_1 \neq CC_2$. Then because of the restrictions regarding $u,v_1$ and $v_2$, we must have $CC_1 = CC_v$ and $CC_2 = CC_u$. Since $CC_2 \neq CC_v$ then it is clear from the figure that for any $z \in V(G)$, $CC_z(v,S,i) \subseteq CC_z(v,S',i)$ implying $\Delta(v,S,i) \leq  \Delta(v,S',i)$, too if $i=0$ or $CC_z(v,S,i) = CC_z(v,S',i)$ implying $\Delta(v,S,i) = \Delta(v,S',i)$ if $i=1$. Notice that in the situation $i=0$ it holds $u \in W(v_1,S,i) = W(v,S,i)$ so $v \in W(u,S,i)$ whereas when $i=1$ then $v \not \in W(u,S,i)$.

\end{proof}

\textbf{The meta-tree simplification.} Finally, with the help of the previous results we can prove that in all these four scenarios the utility for player $u$ with strategy $(S',i)$ is greater than or equal the one for the strategy $(S,i)$. This allows to make a significant simplification.  But before showing the main result of this subsection consider the following observation.

If we suppose that $A(S,0) = \left\{W(u,S,0)\right\}$ then $CC_u(\left\{W(u,S,i)\right\},S,0) = \emptyset$ and therefore $U(u,S,0) \leq 0$. However, $U(u,\emptyset,0) \geq 0$. Therefore:

\begin{remark}
\label{rem:1}
 In order to compute a best response for $u$ it is enough if we consider strategies $(S,i)$ such that $A(S,i) \neq \left\{ W(u,S,i) \right\}$. 
 \end{remark}
 
 Then we are ready to prove the following Lemma:

\begin{lemma}
In all the previous four scenarios it holds $U(u,S',i) \geq U(u,S,i)$. 
\end{lemma}

\begin{proof}  We must distinguish between $u$ being vulnerable or immunised. 

\vskip 5pt 

First of all, if $u$ is immunised then using Lemma \ref{lemm:delete} and Lemma \ref{lemm:swap} we deduce that $\Delta(v,S,i) = \Delta(v,S',i)$ if $v \not \in W(u,S,i)$ implying that $A(S',i) = A(S,i)$. Moreover, by the same two lemmas we know that it holds $CC_u(a,S',i)= CC_u(a,S,i)$ for every $a \in A(S,i) = A(S',i)$, implying in this case $U(u,S',i) \leq U(u,S,i)$.

\vskip 5pt
Now we address the case $u$ vulnerable:

(a) If $\Delta(S,i) = \Delta(u,S,i)$, again, by a direct consequence of Lemma \ref{lemm:delete} and Lemma \ref{lemm:swap} we know that $\Delta(S,i) \leq \Delta(S',i)$.  Therefore we have only two sub-cases to consider.

(a.i) If $\Delta(u,S',i) > \Delta(u,S,i)$. By Lemma \ref{lemm:delete} and Lemma \ref{lemm:swap} we know that $\Delta(v,S,i) = \Delta(v,S',i)$ if $v$ is vulnerable and $v \not \in W(u,S,i)$.  But we can discard the case $A(S,i) = \left\{W(u,S,i)\right\}$ due to the previous remark so $|A(S,i)| > 1$. Then $\Delta(S',i) = \Delta(S,i)$ and combining all these results we reach to $A(S',i) = A(S,i) \setminus \left\{W(u,S,i)\right\}$. In fact, by the same two previous lemmas we know that it holds $CC_u(a,S',i)= CC_u(a,S,i)$ for every $a \in A(S,i) \setminus \left\{W(u,S,i)\right\} = A(S',i)$. Moreover, $CC_u(\left\{W(u,S,i)\right\},S,i) = \emptyset$ because by hypothesis $u$ is vulnerable. Hence:
\begin{gather*}U(u,S',i) - U(u,S,i) = \\
= - \alpha |S'| +\alpha |S|+ \frac{1}{|A(S,i)|-1} \sum_{a \in A(S,i)}|CC_u(a,S,i)| - \frac{1}{|A(S,i)|}\sum_{a \in A(S,i)}|CC_u(a,S,i)| >\\
>-\alpha |S'|+\alpha|S|\end{gather*}

(a.ii) Otherwise, if $\Delta(u,S',i) = \Delta(u,S,i)$ then by Lemma \ref{lemm:delete} and Lemma \ref{lemm:swap} we know that it must hold $W(u,S,i)=W(u,S',i)$. Then, again, by the same lemmas, $\Delta(v,S,i) = \Delta(v,S',i)$ if $v$ is vulnerable and $v \not \in W(u,S,i) = W(u,S',i)$. In conclusion $A(S',i) = A(S,i)$. Furthermore, by considering the same lemmas we know that $CC_u(a,S,i) = CC_u(a,S',i)$ for every $a \in A(S,i) \setminus \left\{W(u,S,i)\right\}=A(S',i) \setminus \left\{W(u,S',i)\right\}$ and $CC_u(a,S,i) \subseteq CC_u(a,S',i)$ for any $a \in W(u,S,i) =W(u,S',i)$. From here we obtain  \begin{equation*}U(u,S',i)-U(u,S,i) \geq -\alpha |S'|+\alpha|S|\end{equation*}

(b) If $\Delta(S,i) < \Delta(u,S,i)$. By Lemma \ref{lemm:delete} and Lemma \ref{lemm:swap} $\Delta(v,S,i) = \Delta(v,S',i)$ for any $v \not \in W(u,S,i)$ and $\Delta(v,S,i) \leq \Delta(v,S',i)$ if $v \in W(u,S,i)$. This implies $A(S,i) = A(S',i)$.  Furthermore, by considering the same lemmas we know that $CC_u(a,S,i) = CC_u(a,S',i)$ for every $a \in A(S,i)=A(S',i)$. Here we obtain \begin{equation*}U(u,S',i)-U(u,S,i) = -\alpha |S'|+\alpha|S|\end{equation*}

Finally, the conclusion is clear because in all the four scenarios $|S| \geq |S'|$. \end{proof}

Now, as a consequence of these results we reach the following corollary:

\begin{corollary}
\label{corol:simplification2}
We can assume without loss of generality that any potential best response $S$ for $u$ satisfies:

- $S$ points only to immunised meta-nodes. 

- $S$ contains at most one link to each component of $V_1(G')$.

- It does not matter towards which immunised meta-node from $V_1(G')$  $S$ is pointing to.  

\end{corollary}

\begin{proof} First, in order to see that we can assume wlog that $S$ points only to immunised meta-nodes apply Lemma \ref{lemm:swap} in scenarios (b.ii) and (b.iii) from the swap-simplification subsection. 

Secondly, in order to see that we can assume wlog that $S$ contains at most one link to each component of $V_1(G'(S,i))$ apply Lemma \ref{lemm:delete} in scenario (a) from the delete-simplification subsection to every link from $S$.

Finally, in order to see that we can change wlog towards which immunised meta-node is $S$ pointing to inside each component from $V_1(G')$ apply Lemma \ref{lemm:swap} in scenario (b.i) from the swap-simplification subsection.

\end{proof}

Consider an immunisation value $i$, and a link-strategy $S$ corresponding to any potential best response strategy for $u$. So far, we can summarise that we have distinguished three distinct levels of abstraction for $S$:

1. The first level of abstraction corresponds to our starting point, writing $S \subseteq V(G)$ meaning that $S$ consists of a collection of nodes from $G(\emptyset,i)$. 

2. The second level of abstraction consists in thinking of $S$ as a collection of distinct meta-nodes from $G'(\emptyset,i)$, writing $S \subseteq V(G')$ meaning that $S$ consists of a collection of meta-nodes from $G'(\emptyset,i)$, since we know that any such link-strategy does not contain two links to distinct nodes from the same meta-node. We reached this level of abstraction after Corollary \ref{corol:simplification1}.  

3. In the last level of abstraction we think $S$ as a subset of meta-tree nodes from $V_1(G'(\emptyset,i))$ and we write $S \subseteq V_1(G')$ meaning that $S$ points to any immunised meta-node inside each of the meta-tree vertices from $S$, since we can assume wlog that any such link-strategy does not contain any link to any vulnerable meta-node from $\mathcal{U}_{\geq 2}(G')$ and it contains, at most, one link to any immunised meta-node from $V_1(G'(\emptyset,i))$ (it does not matter towards which).  We reached this level of abstraction after Corollary \ref{corol:simplification2}.

Taking into the account these three levels of abstraction, in the forthcoming sections, without loss of generality, we always make the next assumption: 

\begin{theorem}
\label{thm:1}
Any potential best response for $u$ can be assumed without loss of generality to be a subset of $V_1(G')$, meaning that the corresponding strategy consists in selecting any immunised meta-node inside each of the components from such subset. 
\end{theorem}

Moreover, given two link-strategies $S,S'$ we consider $S = S'$ as an equality in the third level of abstraction although it could be that $S$ and $S'$ point to distinct endpoints in the first or in the second level of abstraction. 

%

\section{Restricted Strategies and Restricted Utility: A Top-Down Characterisation}

As stated in Theorem \ref{thm:1}, we can assume without loss of generality that any potential best response $(S,i)$ satisfies that $S \subseteq V_1(G')$. 
Despite of this result, there still seems to be an exponential number of possible combinations to explore. For this reason we consider the following idea:

Since we want to find a strategy $(S,i)$ that maximises the value $U(u,S,i)$, we can partition the set of potential best responses into a collection of mutually disjoint subsets depending on the possible values of $\Delta$ and $|A|$. If we know how to find any strategy achieving the maximum restricted utility in each subset, then among such strategies we can pick the one having a maximum utility as a best response for the original problem.

Moreover, we will see that the best response strategy for the objective function when restricted to a sub-tree $T(v)$ can be computed by finding the best response strategies restricted to each of the sub-trees $T_1(v),...,T_{k(v)}(v)$. In this way we can exploit the structure of the meta-tree, breaking the original problem into easier sub-problems to solve.





\subsection{Utility and Restricted Utility}

\label{subsec:strategy}

In this subsection we introduce the partition into which we split the collection of all potential best response strategies. The collection of strategies that we  considering corresponds to the sets of $(\Delta_0,m,T(v))-$strategies where:

(1) The parameter $\Delta_0$ corresponds to the minimum delta value any vulnerable meta-node from the network has. Trivially, $1 \leq \Delta_0 \leq n^2$. 

(2) The subtree $T(v)$ from $T$ in which we are focusing our attention. Recall that we  always consider that $T$ is rooted at the node containing $u$. Trivially $v \in V(T)$. This parameter is really useful because we want to be able to solve the problem for a subtree $T(v)$ by solving the corresponding sub-problems for the subtrees $T_1(v),...,T_{k(v)}(v)$ in which $T(v)$ decomposes. 

(3) The parameter $m$ is the number of meta-nodes from $\mathcal{U}[T(v)]$ achieving a delta value equal to $\Delta_0$.

Before providing the formal definition of a $(\Delta_0,m,T(v))-$strategy, however, we first introduce the concept of restricted attack set.

\begin{definition}

Let $S\subseteq V_1(G')$, $i \in \left\{0,1\right\}$ and $v\in V(T)$. We define the restricted $(\Delta_0,T(v))-$attack set with respect $(S,i)$, using the notation $A(S,i,\Delta_0,T(v))$, as the subset of vulnerable meta-nodes from $T(v)$ having a delta value equal to $\Delta_0$, that is to say, the subset of vulnerable meta-nodes $z   \in \mathcal{U}[T(v)]$ with $\Delta(z,S,i)= \Delta_0$.  

\end{definition}

Now we are ready to define what is a $(\Delta_0,m,T(v))-$strategy. 

\begin{definition}

Let $S\subseteq V_1(G')$, $i \in \left\{0,1\right\}$ and $v\in V(T)$. We say that a strategy $(S,i)$ is a $(\Delta_0,m,T(v))$-strategy iff 

(i) $\Delta(z,S,i) \geq \Delta_0$ for every $z \in \mathcal{U}[G'] \cap T(v)$.

(ii) $m$ equals the number of nodes from $\mathcal{U}[T(v)]$ achieving a delta value equal to $\Delta_0$. 

\end{definition}

Once we have a clear understanding of what is a $(\Delta_0, m,T(v))-$strategy let us introduce the concept of restricted utility:

\begin{definition}
Let $S \subseteq V_1(G')$ and $i\in \left\{0,1\right\}$ and suppose that $(S,i)$ is a $(\Delta_0,m,T(u))-$strategy. We define the restricted $(\Delta_0,|A|,T(v))-$utility of $u$ with respect $S$ as:
\begin{equation*}U(u,S,\Delta_0,|A|,T(v)) = -\alpha |S| + \frac{1}{|A|} \sum_{a \in A(S,1,\Delta_0,T(v)) } |CC_u(a,S,1)|\end{equation*}
\end{definition}

Notice that in the definition of restricted utility we assume that the immunisation value is $i=1$  when we consider the set $A(S,i,\Delta_0,T(v))$ and  the values $|CC_u(a,S,i)|$. 


In the next proposition it is shown that 
there is an intimate relationship between the situations $i=0$ and $i=1$. 

\begin{proposition}
\label{prop:restricted}
Let $(S,i)$ be any potential best-response for player $u$. Let $\Delta_0 = \Delta(S,i)$ and $|A| = |A(S,1,\Delta_0,T(u))|$. 

If $\Delta(S,0) < \Delta(u,\emptyset,0)$ then $i=0$ and 
$U(u,S,i) = U(u,S,\Delta_0,|A|,T(u))$. 

Else, 
$U(u,S,0) = U(u,S,\Delta_0,|A|+1,T(u))$ and $U(u,S,1)=U(u,S,\Delta_0,|A|,T(u))-\beta$.
\end{proposition}

\begin{proof}  One first should notice the following fact: by Theorem \ref{thm:1}, we can assume wlog that all the endpoints of $S$ point to immunised meta-nodes and thus $\Delta(u,S,0) = \Delta(u,\emptyset,0)$. Therefore:

(i) If $\Delta(S,0) < \Delta(u,\emptyset,0)$, then, $\Delta(u,S,0) = \Delta(u,\emptyset,0) >   \Delta(S,0)$ meaning that the meta-node $W(u,S,0)$ is not attacked in this scenario. Similarly, when $i=1$ the meta-node $W(u,S,1)$ is not attacked because $u$ is immunised. Therefore, excluding the meta-nodes $W(u,S,0)$ and $W(u,S,1)$ the corresponding meta-graphs $G'(S,0)$ and $G'(S,1)$ are identical implying $A(S,0) = A(S,1)$ and $CC_u(a,S,0) = CC_u(a,S,1)$ for every $a \in A(S,0)=A(S,1)$, which leads to $U(u,S,0) = U(u,S,1)+\beta$ implying that $(S,1)$ cannot be a best response strategy and, therefore, $i=0$ as we wanted to see. Hence:

\begin{align*} U(u,S,0) =& -\alpha |S|+\frac{1}{|A(S,0)|} \sum_{a \in A(S,0)}|CC_u(a,S,0)| =\\
=& -\alpha |S|+\frac{1}{|A(S,1)|} \sum_{a \in A(S,1)}|CC_u(a,S,1)|=U(u,S,\Delta_0,|A|,T(u))\end{align*}

(ii) If $\Delta(S,0) = \Delta(u,\emptyset,0)$, then, when $i=0$ the meta-node $W(u,S,0)$ is attacked in this scenario whereas when $i=1$ the meta-node $W(u,S,1)$ is not attacked because $u$ is immunised. By remark \ref{rem:1} we can exclude the case $A(S,0) = \left\{ W(u,S,0)\right\}$ and, then, apart from the meta-nodes $W(u,S,0)$ and $W(u,S,1)$ the corresponding meta-graphs $G'(S,0)$ and $G'(S,1)$ are identical implying in this situation that $A(S,1) = A(S,0) \setminus \left\{ W(u,S,0)\right\}$,  $CC_u(a,S,0) = CC_u(a,S,1)$ for every $a \in A(S,1)$ and $CC_u(W(u,S,0),S,0)= \emptyset$. With these results then
\begin{align*} U(u,S,0) &= -\alpha |S| + \frac{1}{|A(S,0)|}\sum_{a \in A(S,0)}|CC_u(a,S,0)|=\\
&= -\alpha |S| + \frac{1}{|A(S,1)|+1}\sum_{a \in A(S,1)}|CC_u(a,S,1)| = U(u,S,\Delta_0,|A|+1,T(u))\end{align*}
And:
\begin{align*}U(u,S,1) =& -\alpha |S| - \beta + \frac{1}{|A(S,1)|} \sum_{a \in A(S,1)}|CC_u(a,S,1)| =\\
=& U(u,S,\Delta_0,|A|,T(u))-\beta\end{align*}
\end{proof}

Hence, in order to compute the best response for $u$ in polynomial time, it is enough to show that the maximum restricted utility can be computed in polynomial time, given that we know the values of the parameters $\Delta_0, |A|$. 

{\bf Remark:} Since the restricted utility function is defined by taking the immunisation value equal to $1$, then in the forthcoming sections  we delete the reference to  the immunisation value, understanding that we are assuming wlog that such value equals $1$. This means that we will write $\Delta(z,S) = \Delta(z,S,1)$, $CC_u(a,S) = CC_u(a,S,1)$, $W(u,S) = W(u,S,1)$, and so on.

\subsection{A Formula for the Restricted Utility}

The aim of this subsection is to obtain a formula for the restricted utility $U(u,S,\Delta_0,|A|,T(v))$ in terms of the restricted utilities $U(u,S_j,\Delta_0,|A|,T(v_j))$ for the corresponding subtrees of $T(v)$. To achieve the formulae given in Proposition \ref{lemm:restrictedutility} we first need to prove some technical lemmas. 

The following three technical lemmas relate the restricted strategies in $T(v)$ with the restricted strategies in the subtrees $T(v_i)$.

\begin{lemma}
\label{lemm:crucial}
Let $S\subseteq V_1(G')$ and $v\in V(T)$. Then $\Delta(z,S) = \Delta(z,S\cap T(v))$ for every $z \in \mathcal{U}[G'] \cap T(v)$. 
\end{lemma}

\begin{proof} If $v = W(u,S)$ the result is trivial. Otherwise, let $S_v = S \cap T(v)$ and $\overline{S_v} = S \cap \overline{T(v)}$ and take $z \in \mathcal{U}[G'] \cap T(v)$. Since $v$ is an articulation point of $T(v)$ then this means that all the endpoints of $S$ belonging to $\overline{T(v)}$ must be in the connected component in which $W(u,S)$ belongs after disconnecting $v$. Then $\Delta(z,S) =\Delta(z,S_v) = \Delta(z,S \cap T(v))$.  

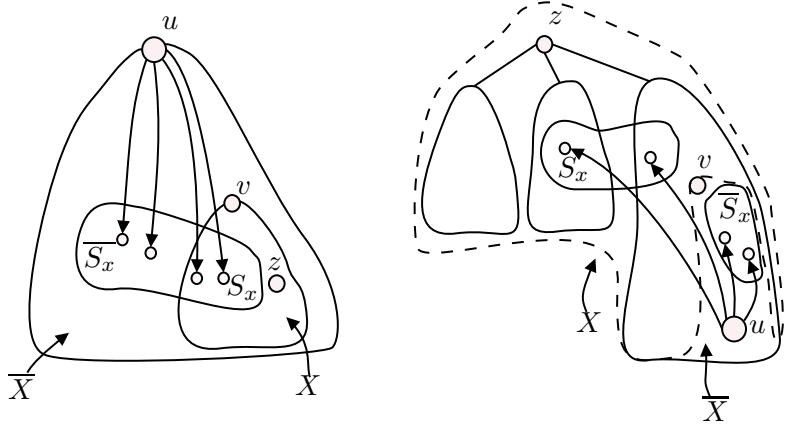
\begin{center}

\tikzset{every picture/.style={line width=0.75pt}} 

\begin{tikzpicture}[x=0.75pt,y=0.75pt,yscale=-0.75,xscale=0.75]

\draw    (141,45.82) .. controls (130.17,77.5) and (130,101.26) .. (125.14,160.28) ;
\draw [shift={(124.91,163)}, rotate = 274.76] [fill={rgb, 255:red, 0; green, 0; blue, 0 }  ][line width=0.08]  [draw opacity=0] (8.93,-4.29) -- (0,0) -- (8.93,4.29) -- cycle    ;
\draw    (146.21,47.35) .. controls (151.88,91.09) and (148.16,124.62) .. (144.2,169.25) ;
\draw [shift={(143.96,172)}, rotate = 275.11] [fill={rgb, 255:red, 0; green, 0; blue, 0 }  ][line width=0.08]  [draw opacity=0] (8.93,-4.29) -- (0,0) -- (8.93,4.29) -- cycle    ;
\draw  [fill={rgb, 255:red, 253; green, 249; blue, 249 }  ,fill opacity=1 ] (121.37,166.85) .. controls (121.37,164.72) and (122.95,163) .. (124.91,163) .. controls (126.87,163) and (128.46,164.72) .. (128.46,166.85) .. controls (128.46,168.98) and (126.87,170.7) .. (124.91,170.7) .. controls (122.95,170.7) and (121.37,168.98) .. (121.37,166.85) -- cycle ;
\draw   (138.19,39) .. controls (139.69,40) and (141,40) .. (146.21,39) .. controls (151.43,38) and (149,29.98) .. (169,40.98) .. controls (189,51.98) and (220,129) .. (247,165) .. controls (274,201) and (272,234) .. (266,239.45) .. controls (260,244.9) and (97.69,253.45) .. (73,244.45) .. controls (48.31,235.45) and (72.71,138.87) .. (84,109.62) .. controls (95.29,80.37) and (101.69,76.34) .. (113.08,61.84) .. controls (124.46,47.34) and (124.3,48.38) .. (128.77,44.23) .. controls (133.25,40.09) and (136.69,38) .. (138.19,39) -- cycle ;
\draw   (203,141) .. controls (212,141.98) and (232,169) .. (236,184) .. controls (240,199) and (247,192) .. (249,220) .. controls (251,248) and (221.33,236) .. (191,240) .. controls (160.67,244) and (156.75,217) .. (169.38,182) .. controls (182,147) and (194,140.02) .. (203,141) -- cycle ;
\draw  [fill={rgb, 255:red, 253; green, 249; blue, 249 }  ,fill opacity=1 ] (140.41,175.85) .. controls (140.41,173.72) and (142,172) .. (143.96,172) .. controls (145.92,172) and (147.51,173.72) .. (147.51,175.85) .. controls (147.51,177.98) and (145.92,179.7) .. (143.96,179.7) .. controls (142,179.7) and (140.41,177.98) .. (140.41,175.85) -- cycle ;
\draw  [fill={rgb, 255:red, 250; green, 242; blue, 242 }  ,fill opacity=1 ] (193.43,142.42) .. controls (193.43,139.33) and (195.79,136.83) .. (198.69,136.83) .. controls (201.6,136.83) and (203.96,139.33) .. (203.96,142.42) .. controls (203.96,145.5) and (201.6,148) .. (198.69,148) .. controls (195.79,148) and (193.43,145.5) .. (193.43,142.42) -- cycle ;
\draw  [fill={rgb, 255:red, 250; green, 242; blue, 242 }  ,fill opacity=1 ] (138.19,39) .. controls (138.19,34.39) and (141.78,30.65) .. (146.21,30.65) .. controls (150.65,30.65) and (154.24,34.39) .. (154.24,39) .. controls (154.24,43.61) and (150.65,47.35) .. (146.21,47.35) .. controls (141.78,47.35) and (138.19,43.61) .. (138.19,39) -- cycle ;
\draw    (152,45.62) .. controls (170.72,69.62) and (176.82,122.37) .. (175,186.08) ;
\draw [shift={(174.91,189)}, rotate = 271.84] [fill={rgb, 255:red, 0; green, 0; blue, 0 }  ][line width=0.08]  [draw opacity=0] (8.93,-4.29) -- (0,0) -- (8.93,4.29) -- cycle    ;
\draw    (154.24,39) .. controls (174.69,54.74) and (186.64,135.51) .. (192.69,186.68) ;
\draw [shift={(192.96,189)}, rotate = 262.87] [fill={rgb, 255:red, 0; green, 0; blue, 0 }  ][line width=0.08]  [draw opacity=0] (8.93,-4.29) -- (0,0) -- (8.93,4.29) -- cycle    ;
\draw  [fill={rgb, 255:red, 253; green, 249; blue, 249 }  ,fill opacity=1 ] (171.37,192.85) .. controls (171.37,190.72) and (172.95,189) .. (174.91,189) .. controls (176.87,189) and (178.46,190.72) .. (178.46,192.85) .. controls (178.46,194.98) and (176.87,196.7) .. (174.91,196.7) .. controls (172.95,196.7) and (171.37,194.98) .. (171.37,192.85) -- cycle ;
\draw  [fill={rgb, 255:red, 253; green, 249; blue, 249 }  ,fill opacity=1 ] (189.41,192.85) .. controls (189.41,190.72) and (191,189) .. (192.96,189) .. controls (194.92,189) and (196.51,190.72) .. (196.51,192.85) .. controls (196.51,194.98) and (194.92,196.7) .. (192.96,196.7) .. controls (191,196.7) and (189.41,194.98) .. (189.41,192.85) -- cycle ;
\draw   (136,144.98) .. controls (145,145.97) and (204,164.75) .. (211,173.75) .. controls (218,182.75) and (217,174.75) .. (219,202.75) .. controls (221,230.75) and (151.33,195.98) .. (121,199.98) .. controls (90.67,203.98) and (90,173.98) .. (98.84,154.85) .. controls (107.68,135.72) and (127,144) .. (136,144.98) -- cycle ;
\draw    (61,256.38) .. controls (63.85,250.68) and (73.02,241.9) .. (86.79,231.2) ;
\draw [shift={(89,229.5)}, rotate = 142.6] [fill={rgb, 255:red, 0; green, 0; blue, 0 }  ][line width=0.08]  [draw opacity=0] (8.93,-4.29) -- (0,0) -- (8.93,4.29) -- cycle    ;
\draw    (251,259.38) .. controls (248.11,250.7) and (253.59,237.9) .. (236.9,220.84) ;
\draw [shift={(235,218.97)}, rotate = 43.45] [fill={rgb, 255:red, 0; green, 0; blue, 0 }  ][line width=0.08]  [draw opacity=0] (8.93,-4.29) -- (0,0) -- (8.93,4.29) -- cycle    ;
\draw   (487,58.38) .. controls (516,59.38) and (550,130) .. (554,145) .. controls (558,160) and (565,205.38) .. (567,233.38) .. controls (569,261.38) and (508.33,244) .. (478,248) .. controls (447.67,252) and (468.38,200) .. (471,162.38) .. controls (473.62,124.77) and (458,57.38) .. (487,58.38) -- cycle ;
\draw   (429,60.38) .. controls (438,61.37) and (444,89) .. (448,104) .. controls (452,119) and (453,112.38) .. (455,140.38) .. controls (457,168.38) and (447.33,157.38) .. (417,161.38) .. controls (386.67,165.38) and (400.62,140.77) .. (401,101.38) .. controls (401.38,62) and (420,59.4) .. (429,60.38) -- cycle ;
\draw   (361.38,63.62) .. controls (370.38,64.6) and (380.38,92) .. (384.38,107) .. controls (388.38,122) and (389,117.38) .. (391,145.38) .. controls (393,173.38) and (379.71,158.62) .. (349.38,162.62) .. controls (319.04,166.62) and (323.38,135.38) .. (336,100.38) .. controls (348.62,65.38) and (352.38,62.63) .. (361.38,63.62) -- cycle ;
\draw  [fill={rgb, 255:red, 250; green, 242; blue, 242 }  ,fill opacity=1 ] (506.43,130.42) .. controls (506.43,127.33) and (508.79,124.83) .. (511.69,124.83) .. controls (514.6,124.83) and (516.96,127.33) .. (516.96,130.42) .. controls (516.96,133.5) and (514.6,136) .. (511.69,136) .. controls (508.79,136) and (506.43,133.5) .. (506.43,130.42) -- cycle ;
\draw  [fill={rgb, 255:red, 250; green, 242; blue, 242 }  ,fill opacity=1 ] (403.43,35.42) .. controls (403.43,32.33) and (405.79,29.83) .. (408.69,29.83) .. controls (411.6,29.83) and (413.96,32.33) .. (413.96,35.42) .. controls (413.96,38.5) and (411.6,41) .. (408.69,41) .. controls (405.79,41) and (403.43,38.5) .. (403.43,35.42) -- cycle ;
\draw    (403.43,35.42) -- (361.38,63.62) ;
\draw    (408.69,41) -- (419,61.38) ;
\draw    (480,60.48) -- (413.96,35.42) ;
\draw  [fill={rgb, 255:red, 250; green, 242; blue, 242 }  ,fill opacity=1 ] (223.43,196.42) .. controls (223.43,193.33) and (225.79,190.83) .. (228.69,190.83) .. controls (231.6,190.83) and (233.96,193.33) .. (233.96,196.42) .. controls (233.96,199.5) and (231.6,202) .. (228.69,202) .. controls (225.79,202) and (223.43,199.5) .. (223.43,196.42) -- cycle ;
\draw  [fill={rgb, 255:red, 250; green, 242; blue, 242 }  ,fill opacity=1 ] (528.19,227) .. controls (528.19,222.39) and (531.78,218.65) .. (536.21,218.65) .. controls (540.65,218.65) and (544.24,222.39) .. (544.24,227) .. controls (544.24,231.61) and (540.65,235.35) .. (536.21,235.35) .. controls (531.78,235.35) and (528.19,231.61) .. (528.19,227) -- cycle ;
\draw    (543,221.38) .. controls (554.5,204.09) and (552.33,196.29) .. (547.05,182.96) ;
\draw [shift={(546,180.33)}, rotate = 68.2] [fill={rgb, 255:red, 0; green, 0; blue, 0 }  ][line width=0.08]  [draw opacity=0] (8.93,-4.29) -- (0,0) -- (8.93,4.29) -- cycle    ;
\draw    (536.21,218.65) .. controls (536.71,206.33) and (537.92,200.82) .. (530.7,172.21) ;
\draw [shift={(530,169.48)}, rotate = 78.01] [fill={rgb, 255:red, 0; green, 0; blue, 0 }  ][line width=0.08]  [draw opacity=0] (8.93,-4.29) -- (0,0) -- (8.93,4.29) -- cycle    ;
\draw    (531,220.48) .. controls (527.69,191.81) and (517.32,158.5) .. (481.65,117.37) ;
\draw [shift={(480,115.48)}, rotate = 48.62] [fill={rgb, 255:red, 0; green, 0; blue, 0 }  ][line width=0.08]  [draw opacity=0] (8.93,-4.29) -- (0,0) -- (8.93,4.29) -- cycle    ;
\draw    (528.19,227) .. controls (515.69,195.69) and (469.88,122.49) .. (428.53,106.4) ;
\draw [shift={(426,105.48)}, rotate = 26.88] [fill={rgb, 255:red, 0; green, 0; blue, 0 }  ][line width=0.08]  [draw opacity=0] (8.93,-4.29) -- (0,0) -- (8.93,4.29) -- cycle    ;
\draw   (445,93.48) .. controls (454,94.47) and (489,80.48) .. (494,91.48) .. controls (499,102.48) and (500,94.38) .. (498,115.38) .. controls (496,136.38) and (445,132.58) .. (430,132.98) .. controls (415,133.38) and (400.16,115.62) .. (409,96.48) .. controls (417.84,77.35) and (436,92.5) .. (445,93.48) -- cycle ;
\draw  [fill={rgb, 255:red, 253; green, 249; blue, 249 }  ,fill opacity=1 ] (418.9,105.48) .. controls (418.9,103.36) and (420.49,101.63) .. (422.45,101.63) .. controls (424.41,101.63) and (426,103.36) .. (426,105.48) .. controls (426,107.61) and (424.41,109.33) .. (422.45,109.33) .. controls (420.49,109.33) and (418.9,107.61) .. (418.9,105.48) -- cycle ;
\draw  [fill={rgb, 255:red, 253; green, 249; blue, 249 }  ,fill opacity=1 ] (476.45,111.63) .. controls (476.45,109.51) and (478.04,107.78) .. (480,107.78) .. controls (481.96,107.78) and (483.55,109.51) .. (483.55,111.63) .. controls (483.55,113.76) and (481.96,115.48) .. (480,115.48) .. controls (478.04,115.48) and (476.45,113.76) .. (476.45,111.63) -- cycle ;
\draw  [fill={rgb, 255:red, 253; green, 249; blue, 249 }  ,fill opacity=1 ] (526.45,165.63) .. controls (526.45,163.51) and (528.04,161.78) .. (530,161.78) .. controls (531.96,161.78) and (533.55,163.51) .. (533.55,165.63) .. controls (533.55,167.76) and (531.96,169.48) .. (530,169.48) .. controls (528.04,169.48) and (526.45,167.76) .. (526.45,165.63) -- cycle ;
\draw  [fill={rgb, 255:red, 253; green, 249; blue, 249 }  ,fill opacity=1 ] (542.45,176.48) .. controls (542.45,174.36) and (544.04,172.63) .. (546,172.63) .. controls (547.96,172.63) and (549.55,174.36) .. (549.55,176.48) .. controls (549.55,178.61) and (547.96,180.33) .. (546,180.33) .. controls (544.04,180.33) and (542.45,178.61) .. (542.45,176.48) -- cycle ;
\draw    (437,215.48) .. controls (433.14,205.93) and (433,203.64) .. (441.99,181.92) ;
\draw [shift={(443,179.48)}, rotate = 112.62] [fill={rgb, 255:red, 0; green, 0; blue, 0 }  ][line width=0.08]  [draw opacity=0] (8.93,-4.29) -- (0,0) -- (8.93,4.29) -- cycle    ;
\draw    (520,272.48) .. controls (514.21,258.97) and (518.66,264.09) .. (517.19,239.3) ;
\draw [shift={(517,236.48)}, rotate = 85.91] [fill={rgb, 255:red, 0; green, 0; blue, 0 }  ][line width=0.08]  [draw opacity=0] (8.93,-4.29) -- (0,0) -- (8.93,4.29) -- cycle    ;
\draw   (534,130) .. controls (543,131) and (544,130) .. (548,147) .. controls (552,164) and (554,179.48) .. (554,182.48) .. controls (554,185.48) and (560,195.08) .. (545,195.48) .. controls (530,195.88) and (512,145) .. (518,140) .. controls (524,135) and (525,129) .. (534,130) -- cycle ;
\draw  [dash pattern={on 4.5pt off 4.5pt}] (561,143.83) .. controls (561,155.83) and (574,222.38) .. (567,233.38) .. controls (560,244.38) and (561.44,197.74) .. (558,187.83) .. controls (554.56,177.93) and (551,144.45) .. (547,135) .. controls (543,125.55) and (536.27,130.76) .. (527,126.83) .. controls (517.73,122.91) and (507.68,123.26) .. (506.43,130.42) .. controls (505.18,137.57) and (528.33,241.83) .. (498,245.83) .. controls (467.67,249.83) and (455,254.83) .. (458,195.83) .. controls (461,136.83) and (327,189.83) .. (323,171.83) .. controls (319,153.83) and (321.22,81) .. (341.61,61) .. controls (362,41) and (371.75,37.26) .. (382,26) .. controls (392.25,14.74) and (389,12) .. (404,8) .. controls (419,4) and (437,14.75) .. (463.25,26.13) .. controls (489.5,37.5) and (489.69,38.16) .. (501.03,44.14) .. controls (512.38,50.13) and (561,131.83) .. (561,143.83) -- cycle ;

\draw (140.31,6) node [anchor=north west][inner sep=0.75pt]   [align=left] {$\displaystyle  \begin{array}{{>{\displaystyle}l}}
u\\
\end{array}$};
\draw (190.31,117) node [anchor=north west][inner sep=0.75pt]   [align=left] {$\displaystyle  \begin{array}{{>{\displaystyle}l}}
v\\
\end{array}$};
\draw (192.96,189) node [anchor=north west][inner sep=0.75pt]   [align=left] {$\displaystyle S_{x}$};
\draw (96.96,166) node [anchor=north west][inner sep=0.75pt]   [align=left] {$\displaystyle \overline{S_{x}}$};
\draw (239,256) node [anchor=north west][inner sep=0.75pt]   [align=left] {$\displaystyle X$};
\draw (46,255) node [anchor=north west][inner sep=0.75pt]   [align=left] {$\displaystyle \overline{X}$};
\draw (500.31,100) node [anchor=north west][inner sep=0.75pt]   [align=left] {$\displaystyle  \begin{array}{{>{\displaystyle}l}}
v\\
\end{array}$};
\draw (400.31,3) node [anchor=north west][inner sep=0.75pt]   [align=left] {$\displaystyle  \begin{array}{{>{\displaystyle}l}}
z\\
\end{array}$};
\draw (220.31,177) node [anchor=north west][inner sep=0.75pt]   [align=left] {$z$};
\draw (535.31,211) node [anchor=north west][inner sep=0.75pt]   [align=left] {$\displaystyle  \begin{array}{{>{\displaystyle}l}}
u\\
\end{array}$};
\draw (428,214) node [anchor=north west][inner sep=0.75pt]   [align=left] {$\displaystyle X$};
\draw (513,270) node [anchor=north west][inner sep=0.75pt]   [align=left] {$\displaystyle \overline{X}$};
\draw (415.45,109.33) node [anchor=north west][inner sep=0.75pt]   [align=left] {$\displaystyle S_{x}$};
\draw (524,132.83) node [anchor=north west][inner sep=0.75pt]   [align=left] {$\displaystyle \overline{S}_{x}$};

\end{tikzpicture}

\end{center}

\end{proof}

\begin{lemma}
\label{lemm:closedbyintersection}
Let $v \in V(T)$, let $S \subseteq V_1(G')$ be a $(\Delta_0,m,T(v))-$strategy and let $m_i = |A(S,\Delta_0,T(v_i))|$. Then $S \cap T(v_i)$ is a $(\Delta_0,m_i,T(v_i))-$strategy.
\end{lemma}

\begin{proof} First, by Lemma \ref{lemm:crucial}, $\Delta(z,S) = \Delta(z,S \cap T(v_i))$ for all $z \in \mathcal{U}[G'] \cap T(v_i)$. Since $S$ is a $(\Delta_0,m,T(v))-$strategy then $\Delta(z,S) \geq \Delta_0$ and, from here, $\Delta(z,S \cap T(v_i)) = \Delta(z,S) \geq \Delta_0$ for every $z \in \mathcal{U}[G'] \cap T(v_i)$. 

Moreover, by definition, $m_i = |A(S,\Delta_0,T(v_i))|$. Now the conclusion is clear.

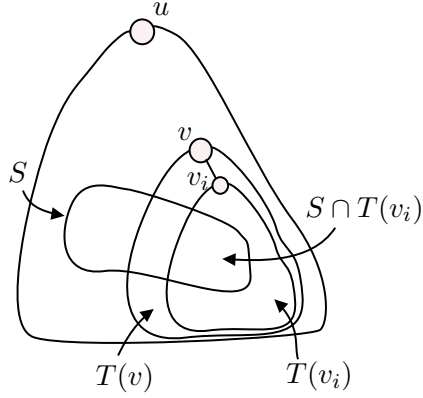
\begin{center}

\tikzset{every picture/.style={line width=0.75pt}} 

\begin{tikzpicture}[x=0.75pt,y=0.75pt,yscale=-0.75,xscale=0.75]

\draw   (320.19,48) .. controls (321.69,49) and (323,49) .. (328.21,48) .. controls (333.43,47) and (331,38.98) .. (351,49.98) .. controls (371,60.98) and (402,138) .. (429,174) .. controls (456,210) and (454,243) .. (448,248.45) .. controls (442,253.9) and (279.69,262.45) .. (255,253.45) .. controls (230.31,244.45) and (254.71,147.87) .. (266,118.62) .. controls (277.29,89.37) and (283.69,85.34) .. (295.08,70.84) .. controls (306.46,56.34) and (306.3,57.38) .. (310.77,53.23) .. controls (315.25,49.09) and (318.69,47) .. (320.19,48) -- cycle ;
\draw   (385,150) .. controls (394,150.98) and (414,178) .. (418,193) .. controls (422,208) and (429,201) .. (431,229) .. controls (433,257) and (403.33,245) .. (373,249) .. controls (342.67,253) and (338.75,226) .. (351.38,191) .. controls (364,156) and (376,149.02) .. (385,150) -- cycle ;
\draw  [fill={rgb, 255:red, 250; green, 242; blue, 242 }  ,fill opacity=1 ] (375.43,151.42) .. controls (375.43,148.33) and (377.79,145.83) .. (380.69,145.83) .. controls (383.6,145.83) and (385.96,148.33) .. (385.96,151.42) .. controls (385.96,154.5) and (383.6,157) .. (380.69,157) .. controls (377.79,157) and (375.43,154.5) .. (375.43,151.42) -- cycle ;
\draw  [fill={rgb, 255:red, 250; green, 242; blue, 242 }  ,fill opacity=1 ] (320.19,48) .. controls (320.19,43.39) and (323.78,39.65) .. (328.21,39.65) .. controls (332.65,39.65) and (336.24,43.39) .. (336.24,48) .. controls (336.24,52.61) and (332.65,56.35) .. (328.21,56.35) .. controls (323.78,56.35) and (320.19,52.61) .. (320.19,48) -- cycle ;
\draw   (318,153.98) .. controls (327,154.97) and (386,173.75) .. (393,182.75) .. controls (400,191.75) and (399,183.75) .. (401,211.75) .. controls (403,239.75) and (333.33,204.98) .. (303,208.98) .. controls (272.67,212.98) and (272,182.98) .. (280.84,163.85) .. controls (289.68,144.72) and (309,153) .. (318,153.98) -- cycle ;
\draw    (312,265.82) .. controls (314.85,260.12) and (320.41,247.33) .. (333.82,236.24) ;
\draw [shift={(336,234.5)}, rotate = 142.6] [fill={rgb, 255:red, 0; green, 0; blue, 0 }  ][line width=0.08]  [draw opacity=0] (8.93,-4.29) -- (0,0) -- (8.93,4.29) -- cycle    ;
\draw    (433,268.38) .. controls (430.11,259.7) and (435.59,246.9) .. (418.9,229.84) ;
\draw [shift={(417,227.97)}, rotate = 43.45] [fill={rgb, 255:red, 0; green, 0; blue, 0 }  ][line width=0.08]  [draw opacity=0] (8.93,-4.29) -- (0,0) -- (8.93,4.29) -- cycle    ;
\draw   (373.14,125.2) .. controls (385.43,126.45) and (412.75,160.75) .. (418.21,179.8) .. controls (423.67,198.84) and (433.23,189.95) .. (435.96,225.5) .. controls (438.7,261.05) and (398.42,248.74) .. (357,253.82) .. controls (315.58,258.9) and (309.98,221.69) .. (327.22,177.26) .. controls (344.46,132.82) and (360.85,123.95) .. (373.14,125.2) -- cycle ;
\draw  [fill={rgb, 255:red, 250; green, 242; blue, 242 }  ,fill opacity=1 ] (360.07,127.86) .. controls (360.07,123.47) and (363.41,119.91) .. (367.53,119.91) .. controls (371.66,119.91) and (375,123.47) .. (375,127.86) .. controls (375,132.26) and (371.66,135.82) .. (367.53,135.82) .. controls (363.41,135.82) and (360.07,132.26) .. (360.07,127.86) -- cycle ;
\draw    (371,133.82) -- (378,147.12) ;
\draw    (248,154.82) .. controls (251.78,164.27) and (258.24,168.36) .. (273.28,171.31) ;
\draw [shift={(276,171.82)}, rotate = 190.01] [fill={rgb, 255:red, 0; green, 0; blue, 0 }  ][line width=0.08]  [draw opacity=0] (8.93,-4.29) -- (0,0) -- (8.93,4.29) -- cycle    ;
\draw    (458,179.82) .. controls (441.77,204.65) and (404.55,202.12) .. (384.69,200.97) ;
\draw [shift={(382,200.82)}, rotate = 3.01] [fill={rgb, 255:red, 0; green, 0; blue, 0 }  ][line width=0.08]  [draw opacity=0] (8.93,-4.29) -- (0,0) -- (8.93,4.29) -- cycle    ;

\draw (324.31,18) node [anchor=north west][inner sep=0.75pt]   [align=left] {$\displaystyle  \begin{array}{{>{\displaystyle}l}}
u\\
\end{array}$};
\draw (347.14,132.2) node [anchor=north west][inner sep=0.75pt]   [align=left] {$\displaystyle  \begin{array}{{>{\displaystyle}l}}
v_{i}\\
\end{array}$};
\draw (235.96,134) node [anchor=north west][inner sep=0.75pt]   [align=left] {$\displaystyle S$};
\draw (423,268) node [anchor=north west][inner sep=0.75pt]   [align=left] {$\displaystyle T( v_{i})$};
\draw (295,269) node [anchor=north west][inner sep=0.75pt]   [align=left] {$\displaystyle T( v)$};
\draw (340.84,101.53) node [anchor=north west][inner sep=0.75pt]   [align=left] {$\displaystyle  \begin{array}{{>{\displaystyle}l}}
v\\
\end{array}$};
\draw (436,156) node [anchor=north west][inner sep=0.75pt]   [align=left] {$\displaystyle S\cap T( v_{i})$};

\end{tikzpicture}

\end{center}

\end{proof}

\begin{lemma}
\label{lemm:crucial2}
Let $S \subseteq V_1(G')$ be a $(\Delta_0,m,T(v))-$strategy and suppose that $a   \in A(S,\Delta_0,T(v_i))$. Then $CC_u(a,S) = CC_u(a,S\cap T(v_i))$. 
\end{lemma}

\begin{proof} $T$ is rooted on $W(u,S)$ and $v_i$ is on the path connecting $a$ with $W(u,S)$, because we are assuming $a \in A(S,\Delta_0,T(v_i))$. Therefore, all the endpoints of $S$ that do not belong to $T(v_i)$ are contained inside the connected component in which $W(u,S)$ belongs when removing $a$. From here the conclusion.

\end{proof}

For a given node $v \in V_1(G')$, any vulnerable meta-node $w\in \mathcal{U}_1(G')$ from $v$ has a delta value independent of the strategy $S$. Let $\widetilde{m}(v,\Delta_0)$ be defined the number of meta-nodes from $\mathcal{U}_1(G')$ contained in $v$ having delta value equal to $\Delta_0$ if $v \in V_1(G')$ and $0$ otherwise. Then, the following formulae will be useful later:

\begin{proposition}
\label{lemm:restrictedutility}
Let $v \in V(T)$,  $S_i$ be $(\Delta_0,m_i,T(v_i))-$strategies for every $i =1,...,k(v)$ and let $S = \cup_{j=1}^{k(v)}S_j$. If $v\in V_{\geq 2}(G') \cap A(S,\Delta_0,T(v))$: 
\begin{equation*}U(u,S,\Delta_0,|A|,T(v)) = \sum_{j=1}^{k(v)} U(u,S_j,\Delta_0,|A|,T(v_j)) + \frac{1}{|A|} \left(|\overline{T(v)}|+\sum_{S_j \neq \emptyset} |T(v_j)| \right)\end{equation*}
Otherwise, 
\begin{equation*}U(u,S,\Delta_0,|A|,T(v)) = \sum_{j=1}^{k(v)} U(u,S_j,\Delta_0,|A|,T(v_j)) + \frac{1}{|A|} \sqrt{ \Delta_0} \cdot \widetilde{m}(v,\Delta_0)\end{equation*}
\end{proposition}

\begin{proof} On the one hand, if $a \in A(S,\Delta_0,T(v_i))$ then we have $CC_u(a,S) = CC_u(a,S_i)$ using Lemma \ref{lemm:crucial2}.  On the other hand:

(i) If $v \in V_{\geq 2}(G')$ and $a=v \in A(S,\Delta_0,T(v))$, then $CC_u(a,S) = \overline{T(v)} \cup \left( \cup_{S_j \neq \emptyset} T(v_j)\right)$. 

(ii) If $v\in V_1(G')$, $a \in \mathcal{U}_1(G')$ with $a \in v$ and $a \in A(S,\Delta_0,T(v))$ as well, then $|CC_u(a,S)| = (n-|a|) =  \sqrt{ \Delta_0}$.

\vskip 5pt

Therefore, suppose first that $v\in V_{\geq 2}(G') \cap A(S,\Delta_0,T(v))$. Then:

\begin{gather*}U(u,S,\Delta_0,|A|,T(v)) =  -\alpha |S| + \frac{1}{|A|}\sum_{a \in A(S,\Delta_0,T(v))}|CC_u(a,S)|=\\
=-\alpha \sum_{i=1}^{k(v)}|S_i|+\frac{1}{|A|}\left(\sum_{i=1}^{k(v)}\sum_{a \in A(S,\Delta_0,T(v_i))}|CC_u(a,S)|\right)+\frac{1}{|A|}|CC_u(v,S)|=\\
=-\alpha \sum_{i=1}^{k(v)}|S_i|+\frac{1}{|A|}\left(\sum_{i=1}^{k(v)}\sum_{a \in A(S,\Delta_0,T(v_i))}|CC_u(a,S_i)|\right)+\frac{1}{|A|}|CC_u(v,S)|=\\
=\sum_{j=1}^{k(v)} U(u,S_j,\Delta_0,|A|,T(v_j)) + \frac{1}{|A|} \left(|\overline{T(v)}|+\sum_{S_j \neq \emptyset} |T(v_j)| \right)\end{gather*}

Otherwise:

\begin{gather*}U(u,S,\Delta_0,|A|,T(v)) =   -\alpha |S| + \frac{1}{|A|}\sum_{a \in A(S,\Delta_0,T(v))}|CC_u(a,S)|=\\
=-\alpha \sum_{i=1}^{k(v)}|S_i|+\frac{1}{|A|}\left(\sum_{i=1}^{k(v)}\sum_{a \in A(S,\Delta_0,T(v_i))}|CC_u(a,S)|\right)+\sum_{a \in \mathcal{U}[v] \land \Delta(a,S) = \Delta_0}\frac{1}{|A|}|CC_u(a,S)|=\\
=-\alpha \sum_{i=1}^{k(v)}|S_i|+\frac{1}{|A|}\left(\sum_{i=1}^{k(v)}\sum_{a \in A(S,\Delta_0,T(v_i))}|CC_u(a,S_i)|\right)+\frac{1}{|A|} \sqrt{\Delta_0} \cdot \widetilde{m}(v,\Delta_0)=\\
=\sum_{j=1}^{k(v)} U(u,S_j,\Delta_0,|A|,T(v_j)) + \frac{1}{|A|}  \sqrt{\Delta_0} \cdot \widetilde{m}(v,\Delta_0) \end{gather*}

\end{proof}

\section{Restricted BR-Strategies: A Bottom-Up Characterisation}

In this section we provide the key results 
in order to design a polynomial time algorithm
that computes  a best response strategy for a given player $u$. 
We introduce the concept of  \emph{restricted BR-strategy} which corresponds somehow to the natural concept of a restricted strategy having the maximum possible restricted utility among all such restricted strategies on a certain sub-tree $T(v)$ from the meta-tree. We show that this key concept has good properties and it allows us to exploit the structure of the meta-tree in an efficient way. We prove that the best response problem on the whole graph we can be  solved
computing the restricted best response problems  for each subtree in a bottom-to-top approach.

\subsection{Restricted BR-strategies on the subtrees}

Recall the definition of restricted strategy. Since we refer to the subtrees $T(v_i)$ it will be convenient to extend such definition as follows: 

\begin{definition}
We say that a strategy  $S \subseteq V_1(G')$ for $u$  is a $(\Delta_0,m,\overline{m},T(v))$-strategy with $\overline{m}=[m_1,...,m_{k(v)}]$ iff $S$ is a $(\Delta_0,m,T(v))$-strategy and $m_i = |A(S,\Delta_0,T(v_i))|$ for each $i=1,...,k(v)$.\end{definition}

Let us define the restricted version of a BR-strategy.

\begin{definition}
We say that a strategy  $S \subseteq V_1(G')$  is a $(\Delta_0,|A|,m,T(v))$-BR-strategy iff $S$ is a $(\Delta_0,m,T(v))$-strategy of maximum restricted $(\Delta_0,|A|,T(v))$-utility among all such $(\Delta_0,m,T(v))$-strategies.
\end{definition}

\begin{definition}
We say that a strategy $S \subseteq V_1(G')$  is a $(\Delta_0,|A|,m,\overline{m},T(v))$-BR-strategy with $\overline{m}=[m_1,...,m_{k(v)}]$ iff $S$ is a $(\Delta_0,m,\overline{m},T(v))$-strategy of maximum restricted $(\Delta_0,|A|,T(v))$-utility over all such $(\Delta_0,m,\overline{m},T(v))$-strategies.
\end{definition}

\begin{definition}
Let $S \subseteq T(v)$ be a strategy for $u$. We say that a pair of meta-tree vertices $(a,b)$ is a $(S,\Delta_0,T(v))$-blocking pair iff $a \in A(S,\Delta_0,T(v))$, $b \in S$ and $b \not \in CC_u(a,\emptyset)$. 
\end{definition}

Notice that for a $(S,\Delta_0,T(v))$-blocking pair $(a,b)$ since $b \in \mathcal{I}(G')$ and $a \in \mathcal{U}(G')$ it follows that we always have $a \neq b$.

\begin{lemma}
\label{lemm:sandwich}
Let $S \subseteq T(v_i)$. Suppose that $\emptyset$ and $S\neq \emptyset$ are $(\Delta_0,m_i,T(v_i))$-strategies. Then, $A(\emptyset, \Delta_0,T(v_i)) = A(S,\Delta_0,T(v_i))$. 
\end{lemma}

\begin{proof} Since $\Delta(z,Z) \leq \Delta(z,Z')$ whenever $Z \subseteq Z'$, then we know that $\Delta(z,\emptyset) \leq \Delta(z,S)$.  Therefore, if $\Delta(z,S) = \Delta_0$ then $\Delta(z,\emptyset) = \Delta_0$ and from here, $A(S,\Delta_0,T(v_i)) \subseteq A(\emptyset, \Delta_0,T(v_i))$. However, by assumption, $|A(\emptyset,\Delta_0,T(v_i))| = |A(S,\Delta_0,T(v_i))| = m_i$. Then the conclusion follows easily. 

\end{proof}

\begin{lemma}
\label{lemm:sandwich2}
Let $S \subseteq T(v_i)$. Suppose that $\emptyset$ and $S\neq \emptyset$ are $(\Delta_0,m_i,T(v_i))$-strategies. Then, any subset $S' \subseteq S$ is a $(\Delta_0,m_i,T(v_i))$-strategy and, more specifically, $A(S',\Delta_0,T(v_i)) = A(S,\Delta_0,T(v_i))$.
\end{lemma}

\begin{proof} Since $\Delta(z,Z) \leq \Delta(z,Z')$ whenever $Z \subseteq Z'$, we know that $\Delta(z,\emptyset) \leq \Delta(z,S') \leq \Delta(z,S)$ for every $z\in \mathcal{U}[T(v_i)]$.  Therefore, $\Delta(z,S') \geq \Delta(z,\emptyset) \geq \Delta_0$ for every $z \in \mathcal{U}[T(v_i)]$. Moreover:

(i) If $z$ verifies $\Delta(z,S') = \Delta_0$ then $\Delta_0 \leq \Delta(z,\emptyset) \leq \Delta(z,S') =\Delta_0$ implying $z\in A(\emptyset,\Delta_0,T(v_i))$. Therefore $A(S',\Delta_0,T(v_i)) \subseteq A(\emptyset,\Delta_0,T(v_i))$.

(ii) If $z$ verifies $\Delta(z,S) = \Delta_0$ then $\Delta_0 \leq \Delta(z,S') \leq \Delta(z,S) = \Delta_0$ implying $z \in A(S',\Delta_0,T(v_i))$. Therefore $A(S,\Delta_0,T(v_i)) \subseteq A(S',\Delta_0,T(v_i))$. 

(iii) By lemma \ref{lemm:sandwich} we know that $A(\emptyset,\Delta_0,T(v_i)) = A(S,\Delta_0,T(v_i))$.

Using (i), (ii) and (iii) we obtain that $A(S',\Delta_0,T(v_i)) = A(\emptyset,\Delta_0,T(v_i)) = A(S,\Delta_0,T(v_i))$. Since $m_i = |A(\emptyset,\Delta_0,T(v_i))|$ we conclude that $|A(S',\Delta_0,T(v_i))| = m_i$ as we wanted to see.

\end{proof}

\begin{lemma}
\label{lemm:blockingpair2}
Let $S \subseteq T(v_i)$. Suppose that $\emptyset$ and $S \neq \emptyset$ are $(\Delta_0,m_i,T(v_i))$-strategies. Then, there cannot exist any $(S,\Delta_0,T(v_i))-$blocking pair.  
\end{lemma}

\begin{proof} Suppose the contrary, then there exists at least one $a \in A(S,\Delta_0,T(v))$ and one $b \in S$ such that $b \not \in CC_u(a,\emptyset)$. Let $S'$ be the strategy that we obtain from $S$ after the removal of all the elements $\tilde{b}$ from $S$ that satisfy $\tilde{b} \not \in CC_u(a,\emptyset)$. Then, since $\Delta(z,Z) \leq \Delta(z,Z')$ whenever $Z \subseteq Z'$, we obtain that $\Delta_0=\Delta(a,\emptyset) \leq \Delta(a,S') < \Delta(a,S) = \Delta_0$, a contradiction. 

\end{proof}

\begin{lemma}
\label{lemm:ccblockingpair} Let $S \subseteq T(v_i)$. If there does not exist any $(S,\Delta_0,T(v_i))-$blocking pair then $CC_u(a,S) = CC_u(a,\emptyset)$ for any $a \in A(S,\Delta_0,T(v_i))$.

\end{lemma} 

\begin{proof} All the endpoints of $S$ belong to the same connected component in which $u$ belongs when removing $a$ from the network. Therefore, $CC_u(a,S) = CC_u(a,\emptyset)$. 

\end{proof}

\begin{corollary}
\label{corol:ccblockingpair}
Let $S \subseteq T(v_i)$. Suppose that $\emptyset$ and $S \neq \emptyset$ are $(\Delta_0,m_i,T(v_i))$-strategies. Then $CC_u(a,S) = CC_u(a,\emptyset)$ for any $a\in A(S,\Delta_0,T(v_i))$.
\end{corollary}

\begin{proposition}  \label{prop:new1}

Let $S$ be a $(\Delta_0,m,\overline{m},T(v))-$strategy with $\overline{m}=[m_1, \ldots, m_{k(v)}]$ and let  $B_i$ be a $(\Delta_0,|A|,m_i,T(v_i))$-BR-strategy. If $|B_i| > 0$, then $|S \cap T(v_i)| > 0$.
\end{proposition}

\begin{proof} 

Let us suppose the contrary, $|B_i| >0$ and $|S \cap T(v_i)| = 0$

Since $S$ is a $(\Delta_0,m,\overline{m},T(v))-$strategy, then $(S \cap T(v_i))=\emptyset$ is a $(\Delta_0,T(v_i),m_i)-$strategy by Lemma \ref{lemm:closedbyintersection}. Therefore, we have $\emptyset$ and $B_i$ with $B_i \neq \emptyset$ are $(\Delta_0,m_i,T(v_i))-$strategies implying:

\begin{itemize}

\item[(i)] $A(\emptyset,\Delta_0,T(v_i)) = A(B_i,\Delta_0,T(v_i))$, by Lemma \ref{lemm:sandwich}.

\item[(ii)] $CC_u(a,B_i) = CC_u(a,\emptyset)$, by Corollary \ref{corol:ccblockingpair}. 

\end{itemize}
Hence, combining (i) with (ii) we have that $U(u, \emptyset, |A|,m_i,T(v_i))>U(u, B_i, |A|,m_i,T(v_i))$ contradicting the fact that $B_i$ is  a $(\Delta_0,|A|,m_i,T(v_i))$-BR-strategy. 






\end{proof}

\begin{proposition} \label{prop:new2}
Let $S$ be a $(\Delta_0,|A|,m,\overline{m},T(v))-$BR-strategy and suppose that $B_i$ is a $(\Delta_0,|A|,m_i,T(v_i))$-BR-strategy. If $|B_i| = 0$ then $|S \cap T(v_i)| \in \left\{0,1\right\}$. 
\end{proposition}

\begin{proof} Suppose the contrary, that $|B_i| = 0$ but $|S \cap T(v_i)| \geq 2$ and we reach a contradiction. First, since $S$ is a $(\Delta_0,|A|,m,\overline{m},T(v))-$strategy then $S \cap T(v_i)$ is a $(\Delta_0,|A|,m_i,T(v_i))$-strategy by Lemma \ref{lemm:closedbyintersection}. But $B_i = \emptyset$ is a $(\Delta_0,m_i,T(v_i))-$strategy as well. Let $b \in S \cap T(v_i)$ and define $S' = S \setminus \left\{b\right\}$.

\begin{itemize}
\item[(i)] $A(S\cap T(v_i),\Delta_0,T(v_i)) = A(S'\cap T(v_i),\Delta_0,T(v_i))$, by Lemma \ref{lemm:sandwich2}.



\item[(ii)] Even more than this, now we see that $CC_u(a,S') = CC_u(a,S)$ for every $a \in A(S, \Delta_0,T(v))$:

\hspace{0.5cm} (ii.a) If $a =v$. Then $CC_u(a,S) = CC_u(a,S')$ because by hypothesis $|S' \cap T(v_i)| \geq 1 >0$ so that $S \cap T(v_i),S' \cap T(v_i) \neq \emptyset$. 

\hspace{0.5cm} (ii.b) If $a \in A(S , \Delta_0,T(v)) \cap (T(v) \setminus (\left\{v\right\} \cup T(v_i)))$ then $CC_u(a,S)  =CC_u(a,S')$ as well, because $b \in T(v_i)$.

\hspace{0.5cm} (ii.c) If $z \in A(S, \Delta_0,T(v)) \cap T(v_i)$, then $CC_u(a,S)  =CC_u(a,S')$ because by Lemma \ref{lemm:blockingpair2} there does not exist any $(S,
\Delta_0,T(v_i))-$blocking pair.

\end{itemize}

Therefore $U(u, S',\Delta_0,|A|,T(v)) = U(u,S,\Delta_0,|A|,T(v))+\alpha > U(u,S,\Delta_0,|A|,T(v))$.

Finally, combining (i) with (ii) we obtain a contradiction with the hypothesis that $S$ is a $(\Delta_0,|A|,m,\overline{m},T(v))-$BR-strategy.  

\end{proof}

\subsection{Composing Restricted BR-Strategies}


Here now we prove the main results that take advantage of the already proven technical lemmas together with the structure of the meta-tree. Since you can think about the meta-tree as a tree having two kind of nodes either from $V_1(G')$ or $V_{\geq 2}(G')$, we need to make a clear distinction between these cases.  

\textbf{The subproblem when $v \in V_1(G')$.} This is the easiest case from the two scenarios that we must consider. Recall that $\widetilde{m}(v,\Delta_0)$ is the number of meta-nodes from $\mathcal{U}_1(G')$ contained in $v$ having delta value equal to $\Delta_0$ if $v \in V_1(G')$ and $0$ otherwise.

\begin{definition}
Suppose that $B_i$ are $(\Delta_0,|A|,m_i,T(v_i))$-BR-strategies for each $i = 1,...,k(v)$ with $m_1+...+m_{k(v)} = m$. Define $S[B_1,...,B_{k(v)};m_1,...,m_{k(v)}] = \cup_{i \geq 1}B_i$. 
\end{definition}

\begin{proposition}
\label{prop:strategy1}
Let $B_j$ be a $(\Delta_0,m_j,T(v_j))-$strategy for $j=1,...,k(v)$ and let $m=m_1+...+m_{k(v)}$. Then $S=S[B_1,...,B_{k(v)};m_1,...,m_{k(v)}]$ is a $(\Delta_0,m+\widetilde{m}(\Delta_0,v),\overline{m},T(v))$-strategy.  
\end{proposition}

\begin{proof} Since $\Delta(z,B) = \Delta(z, B \cap T(v_i)) = \Delta(z,B_i)$ for every $z \in \mathcal{U}[T(v_i)]$ we count exactly $m_i$ attacked meta-nodes from $T(v_i)$. Moreover, the delta value of the vulnerable meta-nodes from $\mathcal{U}[v]$ does not depend on the choice of the strategy. Therefore there are exactly $\widetilde{m}(\Delta_0,v)$ attacked meta-nodes from $v$. The conclusion now follows easily because of the following equality: 
\begin{equation*}|A(S, \Delta_0,T(v))| = m_1+...+m_{k(v)}+\widetilde{m}(\Delta_0,v)\end{equation*}
 
\end{proof} 

Even more than this, we can show that:

\begin{theorem}
\label{THM:1}
Suppose that $m_1,...,m_{k(v)}$ with $m_1+...+m_{k(v)}=m$ are given and $B_j$ is a $(\Delta_0,|A|,m_j,T(v_j))$-BR-strategy for each $j=1,...,k(v)$. Then, $S=S[B_1,...,B_{k(v)};m_1,...,m_{k(v)}]$ is a $(\Delta_0,|A|,m+\widetilde{m}(\Delta_0,v),\overline{m},T(v))-$BR-strategy. \end{theorem}

\begin{proof} By Proposition \ref{prop:strategy1} we know that $S$ is a $(\Delta_0,m+\widetilde{m}(\Delta_0,v),\overline{m},T(v))$-strategy. 

Now, we are going to prove that $S$ corresponds to a strategy with maximum utility for $u$ among all such  $(\Delta_0,m+\widetilde{m}(\Delta_0,v),\overline{m},T(v))$-strategies. In order to prove this we show that $U(u,S,\Delta_0,|A|,T(v)) \geq U(u,S',\Delta_0,|A|,T(v))$ for any  $(\Delta_0,m+\widetilde{m}(\Delta_0,v),\overline{m},T(v))-$strategy $S'$. 
 
Since $S'$ is a $(\Delta_0,m+\widetilde{m}(\Delta_0,v),\overline{m},T(v))-$strategy then by Lemma \ref{lemm:closedbyintersection} $S' \cap T(v_j)$ is a $(\Delta_0,m_j,T(v_j))-$strategy and, therefore, by definition of best response we must have $U(u,S' \cap T(v_j),\Delta_0,|A|,T(v_j)) \leq U(u,B_j,\Delta_0,|A|,T(v_j))$ for every $j \in \left\{1,...,k(v)\right\}$. 
 
 Then using this together with the formula from Proposition \ref{lemm:restrictedutility} we obtain the following:
 \begin{equation*}U(u,S',\Delta_0,|A|,T(v)) = \end{equation*}
 \begin{equation*}=\sum_{j = 1}^{k(v)}U(u,S' \cap T(v_j), \Delta_0,|A|,T(v_j))+ \frac{1}{|A|} \Delta_0 \cdot \widetilde{m}(\Delta_0,v) \leq\end{equation*}
  \begin{equation*}\leq  \sum_{j = 1}^{k(v)}U(u,B_j, \Delta_0,|A|,T(v_j))+ \frac{1}{|A|} \Delta_0 \cdot \widetilde{m}(\Delta_0,v)|\end{equation*}
  
  But by construction of $S$:
  \begin{equation*} \sum_{j = 1}^{k(v)}U(u,B_j, \Delta_0,|A|,T(v_j))+ \frac{1}{|A|} \Delta_0 \cdot \widetilde{m}(\Delta_0,v) =\end{equation*}
  \begin{equation*}=  \sum_{j = 1}^{k(v)}U(u,S \cap T(v_j), \Delta_0,|A|,T(v_j))+ \frac{1}{|A|} \Delta_0 \cdot \widetilde{m}(\Delta_0,v)=\end{equation*}
  \begin{equation*}=U(u,S,\Delta_0,|A|,T(v))\end{equation*}
  
  And now the conclusion follows easily.

\end{proof}


\textbf{The subproblem when $v \in V_{\geq 2}(G')$.} In contrast to the previous case which was easier, this scenario gets a little more involved. 

\begin{proposition} \label{prop:new3}
For a given $v\in V_{\geq 2}(G')$, let $S$ be a $(\Delta_0,|A|,m,\overline{m},T(v))-$BR-strategy and $B_i$ a $(\Delta_0,|A|,m_i,T(v_i))-$BR-strategy with $B_i \neq \emptyset$. Then $S' = (S \setminus T(v_i)) \cup B_i$ is a $(\Delta_0,|A|,m,\overline{m},T(v))-$BR-strategy as well.  
\end{proposition}

\begin{proof} First, we claim that $S'$ is a $(\Delta_0,m,\overline{m},T(v))-$strategy. This is because: 

(i) $\Delta(z,S') = \Delta(z,S' \cap T(v_i)) = \Delta(z,B_i)$ for every $z \in \mathcal{U}[T(v_i)]$, using Lemma \ref{lemm:crucial}. Therefore, $|A(S',\Delta_0,T(v_i))|=|A(B_i,\Delta_0,T(v_i))| = m_i$, too. 

(ii) $\Delta(z,S') = \Delta(z,S)$ for every $z \in \mathcal{U}[T(v_j)]$ with $j \neq i$ because the subset of nodes from $S'$ that are not nodes from $S$ together with the set of nodes from $S$ that are not nodes from $S'$ are a subset of $T(v_i)$ which clearly does not belong to $T(v_j)$. Therefore, $|A(S',\Delta_0,T(v_j))|=|A(B_j,\Delta_0,T(v_j))| = m_j$ for every $j \neq i$.

(iii) Finally, $\Delta(v,S') = \Delta(v,S)$ because by hypothesis, $B_i \neq \emptyset$ implying $S \cap T(v_i) \neq \emptyset$ as well by Proposition \ref{prop:new1}. 
 
\vskip 5pt

By (i)+(ii) we deduce that $A(S \cap T(v_j),\Delta_0,T(v_j)) = A(S' \cap T(v_j),\Delta_0,T(v_j))$ for each $j = 1,..., k(v)$ and together with (iii) then we conclude $A(S,\Delta_0,T(v)) = A(S', \Delta_0,T(v))$. 

\vskip 5pt

Furthermore, we claim that $CC_u(a,S')=CC_u(a,S)$ for every $a \in A(S,   \Delta_0, T(v)) \setminus A(S \cap T(v_i),\Delta_0,T(v_i))$. First, notice that $CC_u(v,S') = CC_u(v,S)$ if $v \in A(S, \Delta_0, T(v))$ because $B_i \neq \emptyset$ implies $S\cap T(v_i) \neq \emptyset$ as well using Proposition \ref{prop:new1}. Moreover, $CC_u(a,S) = CC_u(a,S\cap T(v_j)) = CC_u(a,S' \cap T(v_j)) =CC_u(a,S')$ for every $a \in A(S,\Delta_0,T(v_j))$ with $j \neq i$ using Lemma \ref{lemm:crucial2}. In this way:


\begin{equation*}U(u,S',\Delta_0,|A|,T(v)) - U(u,S,\Delta_0,|A|,T(v)) = \end{equation*}
\begin{equation*}= -\alpha |S'|+\alpha |S| + \frac{1}{|A|}\left(\sum_{a \in A(S',\Delta_0, T(v))}|CC_u(a,S')| - \sum_{a \in A(S,\Delta_0,T(v))}|CC_u(a,S)|\right)=\end{equation*}
\begin{equation*} =- \alpha |B_i| + \alpha |S \cap T(v_i)| + \frac{1}{|A|}\left(\sum_{a \in A(S',\Delta_0, T(v_i))}|CC_u(a,S')|-\sum_{a \in A(S,\Delta_0, T(v_i))}|CC_u(a,S)|\right) =\end{equation*}
\begin{equation*}=U(u,B_i,\Delta_0,|A|,T(v_i)) - U(u,S\cap T(v_i),\Delta_0,|A|,T(v_i)) \geq 0 \end{equation*}

Finally, since $S$ is a $(\Delta_0,|A|,m,\overline{m},T(v))-$BR-strategy we conclude that $U(u,S',\Delta_0,|A|,T(v)) = U(u,S,\Delta_0,|A|,T(v))$. 

Therefore, $S'$ has maximum $(\Delta_0,|A|,T(v))-$restricted utility over all $(\Delta_0,m,\overline{m},T(v))-$strategies, as we wanted to see.

\end{proof}

\begin{definition}
Suppose that $B_i$ are $(\Delta_0,|A|,m_i,T(v_i))$-BR-strategies for each $i = 1,...,k(v)$ with $m_1+...+m_{k(v)} = m$. Let $I \subseteq [k(v)]$ be the subset of indices $i$ for which $B_i = \emptyset$ and let $\overline{I} = [k(v)] \setminus I$.

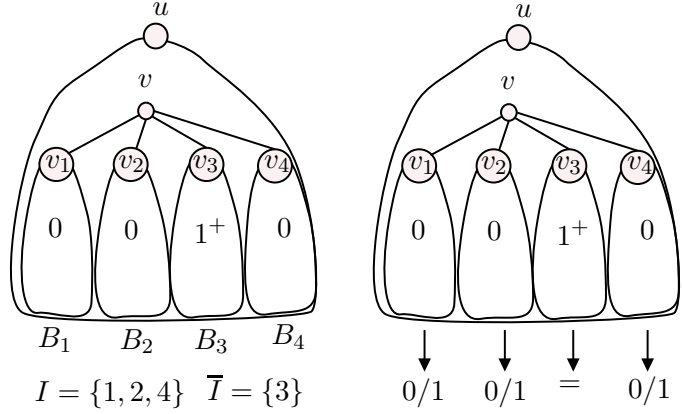
\begin{center}

\tikzset{every picture/.style={line width=0.75pt}} 

\begin{tikzpicture}[x=0.75pt,y=0.75pt,yscale=-0.75,xscale=0.75]

\draw   (174,35) .. controls (184.84,30) and (210,38) .. (220,48) .. controls (230,58) and (268.97,81.72) .. (282.97,118.9) .. controls (296.97,156.08) and (299,211.85) .. (293,217.3) .. controls (287,222.75) and (117,227.75) .. (99,218.3) .. controls (81,208.85) and (96.97,123.62) .. (107,101.85) .. controls (117.03,80.08) and (118,71.85) .. (136,60) .. controls (154,48.15) and (163.16,40) .. (174,35) -- cycle ;
\draw   (223,110) .. controls (232,110.98) and (238,134.82) .. (242,149.82) .. controls (246,164.82) and (247,187.3) .. (247,206.3) .. controls (247,225.3) and (231,218.3) .. (214,221.3) .. controls (197,224.3) and (194.62,192.63) .. (201,154.82) .. controls (207.38,117) and (214,109.02) .. (223,110) -- cycle ;
\draw  [fill={rgb, 255:red, 250; green, 242; blue, 242 }  ,fill opacity=1 ] (180.19,32) .. controls (180.19,27.39) and (183.78,23.65) .. (188.21,23.65) .. controls (192.65,23.65) and (196.24,27.39) .. (196.24,32) .. controls (196.24,36.61) and (192.65,40.35) .. (188.21,40.35) .. controls (183.78,40.35) and (180.19,36.61) .. (180.19,32) -- cycle ;
\draw   (173,112) .. controls (182,112.98) and (188,136.82) .. (192,151.82) .. controls (196,166.82) and (198,188.3) .. (198,207.3) .. controls (198,226.3) and (180,218.3) .. (163,221.3) .. controls (146,224.3) and (144.62,194.63) .. (151,156.82) .. controls (157.38,119) and (164,111.02) .. (173,112) -- cycle ;
\draw   (123,111) .. controls (132,111.98) and (138,135.82) .. (142,150.82) .. controls (146,165.82) and (146,194.3) .. (145,211.3) .. controls (144,228.3) and (119,216.78) .. (109,218.3) .. controls (99,219.82) and (94.62,193.63) .. (101,155.82) .. controls (107.38,118) and (114,110.02) .. (123,111) -- cycle ;
\draw  [fill={rgb, 255:red, 250; green, 242; blue, 242 }  ,fill opacity=1 ] (110,118.4) .. controls (110,112.22) and (115.1,107.22) .. (121.39,107.22) .. controls (127.69,107.22) and (132.79,112.22) .. (132.79,118.4) .. controls (132.79,124.58) and (127.69,129.58) .. (121.39,129.58) .. controls (115.1,129.58) and (110,124.58) .. (110,118.4) -- cycle ;
\draw   (272,111) .. controls (281,111.98) and (287,135.82) .. (291,150.82) .. controls (295,165.82) and (296,182.82) .. (296,201.82) .. controls (296,220.82) and (280,216.3) .. (263,219.3) .. controls (246,222.3) and (247,193.3) .. (250,155.82) .. controls (253,118.33) and (263,110.02) .. (272,111) -- cycle ;
\draw  [fill={rgb, 255:red, 250; green, 242; blue, 242 }  ,fill opacity=1 ] (176.43,82.42) .. controls (176.43,79.33) and (178.79,76.83) .. (181.69,76.83) .. controls (184.6,76.83) and (186.96,79.33) .. (186.96,82.42) .. controls (186.96,85.5) and (184.6,88) .. (181.69,88) .. controls (178.79,88) and (176.43,85.5) .. (176.43,82.42) -- cycle ;
\draw  [fill={rgb, 255:red, 250; green, 242; blue, 242 }  ,fill opacity=1 ] (160,118.4) .. controls (160,112.22) and (165.1,107.22) .. (171.39,107.22) .. controls (177.69,107.22) and (182.79,112.22) .. (182.79,118.4) .. controls (182.79,124.58) and (177.69,129.58) .. (171.39,129.58) .. controls (165.1,129.58) and (160,124.58) .. (160,118.4) -- cycle ;
\draw  [fill={rgb, 255:red, 250; green, 242; blue, 242 }  ,fill opacity=1 ] (211,118.4) .. controls (211,112.22) and (216.1,107.22) .. (222.39,107.22) .. controls (228.69,107.22) and (233.79,112.22) .. (233.79,118.4) .. controls (233.79,124.58) and (228.69,129.58) .. (222.39,129.58) .. controls (216.1,129.58) and (211,124.58) .. (211,118.4) -- cycle ;
\draw  [fill={rgb, 255:red, 250; green, 242; blue, 242 }  ,fill opacity=1 ] (257,119.4) .. controls (257,113.22) and (262.1,108.22) .. (268.39,108.22) .. controls (274.69,108.22) and (279.79,113.22) .. (279.79,119.4) .. controls (279.79,125.58) and (274.69,130.58) .. (268.39,130.58) .. controls (262.1,130.58) and (257,125.58) .. (257,119.4) -- cycle ;
\draw    (177,85.22) -- (129,111.22) ;
\draw    (181.69,88) -- (175,108.22) ;
\draw    (186.96,82.42) -- (268.39,108.22) ;
\draw    (186,86.22) -- (222.39,107.22) ;
\draw   (418,36) .. controls (428.84,31) and (454,39) .. (464,49) .. controls (474,59) and (512.97,82.72) .. (526.97,119.9) .. controls (540.97,157.08) and (543,212.85) .. (537,218.3) .. controls (531,223.75) and (361,228.75) .. (343,219.3) .. controls (325,209.85) and (340.97,124.62) .. (351,102.85) .. controls (361.03,81.08) and (362,72.85) .. (380,61) .. controls (398,49.15) and (407.16,41) .. (418,36) -- cycle ;
\draw   (467,111) .. controls (476,111.98) and (482,135.82) .. (486,150.82) .. controls (490,165.82) and (491,188.3) .. (491,207.3) .. controls (491,226.3) and (475,219.3) .. (458,222.3) .. controls (441,225.3) and (438.62,193.63) .. (445,155.82) .. controls (451.38,118) and (458,110.02) .. (467,111) -- cycle ;
\draw  [fill={rgb, 255:red, 250; green, 242; blue, 242 }  ,fill opacity=1 ] (424.19,33) .. controls (424.19,28.39) and (427.78,24.65) .. (432.21,24.65) .. controls (436.65,24.65) and (440.24,28.39) .. (440.24,33) .. controls (440.24,37.61) and (436.65,41.35) .. (432.21,41.35) .. controls (427.78,41.35) and (424.19,37.61) .. (424.19,33) -- cycle ;
\draw   (417,113) .. controls (426,113.98) and (432,137.82) .. (436,152.82) .. controls (440,167.82) and (442,189.3) .. (442,208.3) .. controls (442,227.3) and (424,219.3) .. (407,222.3) .. controls (390,225.3) and (388.62,195.63) .. (395,157.82) .. controls (401.38,120) and (408,112.02) .. (417,113) -- cycle ;
\draw   (367,112) .. controls (376,112.98) and (382,136.82) .. (386,151.82) .. controls (390,166.82) and (390,195.3) .. (389,212.3) .. controls (388,229.3) and (363,217.78) .. (353,219.3) .. controls (343,220.82) and (338.62,194.63) .. (345,156.82) .. controls (351.38,119) and (358,111.02) .. (367,112) -- cycle ;
\draw  [fill={rgb, 255:red, 250; green, 242; blue, 242 }  ,fill opacity=1 ] (354,119.4) .. controls (354,113.22) and (359.1,108.22) .. (365.39,108.22) .. controls (371.69,108.22) and (376.79,113.22) .. (376.79,119.4) .. controls (376.79,125.58) and (371.69,130.58) .. (365.39,130.58) .. controls (359.1,130.58) and (354,125.58) .. (354,119.4) -- cycle ;
\draw   (516,112) .. controls (525,112.98) and (531,136.82) .. (535,151.82) .. controls (539,166.82) and (540,183.82) .. (540,202.82) .. controls (540,221.82) and (524,217.3) .. (507,220.3) .. controls (490,223.3) and (491,194.3) .. (494,156.82) .. controls (497,119.33) and (507,111.02) .. (516,112) -- cycle ;
\draw  [fill={rgb, 255:red, 250; green, 242; blue, 242 }  ,fill opacity=1 ] (420.43,83.42) .. controls (420.43,80.33) and (422.79,77.83) .. (425.69,77.83) .. controls (428.6,77.83) and (430.96,80.33) .. (430.96,83.42) .. controls (430.96,86.5) and (428.6,89) .. (425.69,89) .. controls (422.79,89) and (420.43,86.5) .. (420.43,83.42) -- cycle ;
\draw  [fill={rgb, 255:red, 250; green, 242; blue, 242 }  ,fill opacity=1 ] (404,119.4) .. controls (404,113.22) and (409.1,108.22) .. (415.39,108.22) .. controls (421.69,108.22) and (426.79,113.22) .. (426.79,119.4) .. controls (426.79,125.58) and (421.69,130.58) .. (415.39,130.58) .. controls (409.1,130.58) and (404,125.58) .. (404,119.4) -- cycle ;
\draw  [fill={rgb, 255:red, 250; green, 242; blue, 242 }  ,fill opacity=1 ] (455,119.4) .. controls (455,113.22) and (460.1,108.22) .. (466.39,108.22) .. controls (472.69,108.22) and (477.79,113.22) .. (477.79,119.4) .. controls (477.79,125.58) and (472.69,130.58) .. (466.39,130.58) .. controls (460.1,130.58) and (455,125.58) .. (455,119.4) -- cycle ;
\draw  [fill={rgb, 255:red, 250; green, 242; blue, 242 }  ,fill opacity=1 ] (501,120.4) .. controls (501,114.22) and (506.1,109.22) .. (512.39,109.22) .. controls (518.69,109.22) and (523.79,114.22) .. (523.79,120.4) .. controls (523.79,126.58) and (518.69,131.58) .. (512.39,131.58) .. controls (506.1,131.58) and (501,126.58) .. (501,120.4) -- cycle ;
\draw    (421,86.22) -- (373,112.22) ;
\draw    (425.69,89) -- (419,109.22) ;
\draw    (430.96,83.42) -- (512.39,109.22) ;
\draw    (430,87.22) -- (466.39,108.22) ;
\draw    (369,229.87) -- (369,254.87) ;
\draw [shift={(369,257.87)}, rotate = 270] [fill={rgb, 255:red, 0; green, 0; blue, 0 }  ][line width=0.08]  [draw opacity=0] (8.93,-4.29) -- (0,0) -- (8.93,4.29) -- cycle    ;
\draw    (423,229.87) -- (423,254.87) ;
\draw [shift={(423,257.87)}, rotate = 270] [fill={rgb, 255:red, 0; green, 0; blue, 0 }  ][line width=0.08]  [draw opacity=0] (8.93,-4.29) -- (0,0) -- (8.93,4.29) -- cycle    ;
\draw    (469,229.87) -- (469,254.87) ;
\draw [shift={(469,257.87)}, rotate = 270] [fill={rgb, 255:red, 0; green, 0; blue, 0 }  ][line width=0.08]  [draw opacity=0] (8.93,-4.29) -- (0,0) -- (8.93,4.29) -- cycle    ;
\draw    (519,229.87) -- (519,254.87) ;
\draw [shift={(519,257.87)}, rotate = 270] [fill={rgb, 255:red, 0; green, 0; blue, 0 }  ][line width=0.08]  [draw opacity=0] (8.93,-4.29) -- (0,0) -- (8.93,4.29) -- cycle    ;

\draw (175.31,-1) node [anchor=north west][inner sep=0.75pt]   [align=left] {$\displaystyle  \begin{array}{{>{\displaystyle}l}}
u\\
\end{array}$};
\draw (107,226) node [anchor=north west][inner sep=0.75pt]   [align=left] {$\displaystyle B_{1}$};
\draw (112.14,110.2) node [anchor=north west][inner sep=0.75pt]   [align=left] {$\displaystyle v_{1}$};
\draw (162,227) node [anchor=north west][inner sep=0.75pt]   [align=left] {$\displaystyle B_{2}$};
\draw (212,227) node [anchor=north west][inner sep=0.75pt]   [align=left] {$\displaystyle B_{3}$};
\draw (264,224) node [anchor=north west][inner sep=0.75pt]   [align=left] {$\displaystyle B_{4}$};
\draw (174.14,56.2) node [anchor=north west][inner sep=0.75pt]   [align=left] {$\displaystyle v$};
\draw (161.14,110.2) node [anchor=north west][inner sep=0.75pt]   [align=left] {$\displaystyle v_{2}$};
\draw (211.14,110.2) node [anchor=north west][inner sep=0.75pt]   [align=left] {$\displaystyle v_{3}$};
\draw (259.14,110.2) node [anchor=north west][inner sep=0.75pt]   [align=left] {$\displaystyle v_{4}$};
\draw (114,153) node [anchor=north west][inner sep=0.75pt]   [align=left] {$\displaystyle 0$};
\draw (165,154) node [anchor=north west][inner sep=0.75pt]   [align=left] {$\displaystyle 0$};
\draw (268,153) node [anchor=north west][inner sep=0.75pt]   [align=left] {$\displaystyle 0$};
\draw (212,152.82) node [anchor=north west][inner sep=0.75pt]   [align=left] {$\displaystyle 1^{+}$};
\draw (105,261) node [anchor=north west][inner sep=0.75pt]   [align=left] {$\displaystyle I=\{1,2,4\}$};
\draw (221,258) node [anchor=north west][inner sep=0.75pt]   [align=left]  {$\displaystyle \overline{I}=\{3\}$ };
\draw (419.31,0) node [anchor=north west][inner sep=0.75pt]   [align=left] {$\displaystyle  \begin{array}{{>{\displaystyle}l}}
u\\
\end{array}$};
\draw (356.14,111.2) node [anchor=north west][inner sep=0.75pt]   [align=left] {$\displaystyle v_{1}$};
\draw (407,260) node [anchor=north west][inner sep=0.75pt]   [align=left] {$\displaystyle 0/1$};
\draw (418.14,57.2) node [anchor=north west][inner sep=0.75pt]   [align=left] {$\displaystyle v$};
\draw (405.14,111.2) node [anchor=north west][inner sep=0.75pt]   [align=left] {$\displaystyle v_{2}$};
\draw (455.14,111.2) node [anchor=north west][inner sep=0.75pt]   [align=left] {$\displaystyle v_{3}$};
\draw (503.14,111.2) node [anchor=north west][inner sep=0.75pt]   [align=left] {$\displaystyle v_{4}$};
\draw (358,154) node [anchor=north west][inner sep=0.75pt]   [align=left] {$\displaystyle 0$};
\draw (409,155) node [anchor=north west][inner sep=0.75pt]   [align=left] {$\displaystyle 0$};
\draw (512,154) node [anchor=north west][inner sep=0.75pt]   [align=left] {$\displaystyle 0$};
\draw (456,153.82) node [anchor=north west][inner sep=0.75pt]   [align=left] {$\displaystyle 1^{+}$};
\draw (352,260) node [anchor=north west][inner sep=0.75pt]   [align=left] {$\displaystyle 0/1$};
\draw (504,259) node [anchor=north west][inner sep=0.75pt]   [align=left] {$\displaystyle 0/1$};
\draw (457,262) node [anchor=north west][inner sep=0.75pt]   [align=left] {$\displaystyle =$};

\end{tikzpicture}

\end{center}


Let $J_= \subseteq I$ be such that $-\alpha |J_=| + \frac{1}{|A|} \sum_{j \in J_=} |T(v_j)|$ is maximum among all subsets $Z \subseteq I$ verifying: 
\begin{equation*} \Delta(v, \left(\cup_{i \in \overline{I}} B_i \right) \cup \left( \cup_{j \in Z} \left\{ v_j \right\} \right) ) = \Delta_0 \end{equation*}

Similarly, let $J_> \subseteq I$ be such that $-\alpha |J_>| + \frac{1}{|A|} \sum_{j \in J_>} |T(v_j)|$ is maximum among all subsets $Z \subseteq I$ verifying: 
\begin{equation*} \Delta(v, \left(\cup_{i \in \overline{I}} B_i \right) \cup \left( \cup_{j \in Z} \left\{ v_j \right\} \right) ) > \Delta_0 \end{equation*}

Finally, define $S_=[B_1,...,B_{k(v)};m_1,...,m_{k(v)}] = \left( \cup_{i \in \overline{I}} B_i \right) \cup \left( \cup_{j \in J_=} \left\{v_j \right\} \right)$ if such subset $J_=$ exists and $S_>[B_1,...,B_{k(v)};m_1,...,m_{k(v)}] = \left( \cup_{i \in \overline{I}} B_i \right) \cup \left( \cup_{j \in J_>} \left\{v_j \right\} \right)$ if such subset $J_>$ exists.
\end{definition}

\begin{proposition}
\label{prop:strategy2}
Let $B_j$ be a $(\Delta_0,m_j,T(v_j))-$strategy for $j=1,...,k(v)$ and let $m=m_1+...+m_{k(v)}$. Then $S_=[B_1,...,B_{k(v)};m_1,...,m_{k(v)}]$ is a $(\Delta_0,m+1,\overline{m},T(v))$-strategy and $S_>[B_1,...,B_{k(v)};m_1,...,m_{k(v)}]$ is a $(\Delta_0,m,\overline{m},T(v))-$strategy otherwise.  
\end{proposition}

\begin{proof}   Let $x$ denote the symbol $=$ or $>$ and let $S=S_x[B_1,...,B_{k(v)};m_1,...,m_{k(v)}]$. If $i \in \overline{I} \cup (I \setminus J_x) $ and $z \in \mathcal{U}[T(v_i)]$, then using Lemma \ref{lemm:crucial} $\Delta(z,S) = \Delta(z, S \cap T(v_i)) = \Delta(z,B_i)$ so that $|A(S,\Delta_0,T(v_i))| = m_i$. Suppose now that $i \in J_x$. Since the only links to $B_i$ that we are adding in this case is the link to the node $v_i$, then it is clear, using Lemma \ref{lemm:crucial} that $\Delta(z,S) = \Delta(z,S \cap T(v_i)) = \Delta(z,B_i \cup \left\{v_i\right\}) = \Delta(z,B_i)$ for any $z \in \mathcal{U}[T(v_i) \setminus \left\{v_i\right\}]$. 

Furthermore, when $z \in \mathcal{U}_1(G')$ is a vulnerable meta-node from $v_i \in V_1(G')$ with $i \in J_x$, then $\Delta(z,S) = \Delta(z,\emptyset)$. Therefore, the number of vulnerable meta-nodes from $T(v_i)$ achieving a delta value equal to $\Delta_0$ is $|A(B_i,\Delta_0,T(v_i)\setminus \left\{v_i\right\})|+\widetilde{m}(v_i,\Delta_0) =|A(B_i,\Delta_0,T(v_i))|= m_i$. 

Finally, by construction, the delta value of $v$ satisfies that it is equal to $\Delta_0$ if $x$ equals $=$ and therefore in this case the cardinality of the restricted attack set $A(S,\Delta_0,T(v))$ is one unit larger than $m$ or greater than $\Delta_0$ otherwise and then in this case the cardinality of the restricted attack set $A(S,\Delta_0,T(v))$ equals exactly $m$. This completes the proof of the result. 

\end{proof}

\begin{theorem}
\label{thm:2}
Suppose that $m_1,...,m_{k(v)}$ with $m_1+...+m_{k(v)}=m$ are given and $B_j$ is a $(\Delta_0,|A|,m_j,T(v_j))$-BR-strategy for each $j=1,...,k(v)$. Also, let $S_x=S_x[B_1,...,B_{k(v)};m_1,...,m_{k(v)}]$ with $x \in \left\{ =, > \right\}$. Then, $S_=$ maximises $U(u,\cdot,\Delta_0,|A|,m+1,T(v))$ over all $(\Delta_0,m+1,\overline{m},T(v))-$strategies and $S_>$ maximises $U(u,\cdot, \Delta_0,|A|,m,T(v))$ over all $(\Delta_0,m,\overline{m},T(v))-$strategies. 
\end{theorem}

\begin{proof}  We start proving the case $x$ equals $=$, the other case $x$ equals $>$ can be proved analogously. 

Let $S = S_=$. We show that $U(u,S,\Delta_0,|A|,T(v))$ is the maximum utility that any $(\Delta_0,|A|,m+1,\overline{m},T(v))-$strategy can attain. In order to do this let $S'$ be a $(\Delta_0,|A|,m+1,\overline{m},T(v))-$BR-strategy and then we show that $U(u,S,\Delta_0,|A|,T(v))=U(u,S',\Delta_0,|A|,T(v))$.

Since $S'$ is a $(\Delta_0,|A|,m+1,\overline{m},T(v))$-BR-strategy, and $B_i$ are $(\Delta_0,|A|,m_i,T(v_i))-$BR-strategies then $S'' = \cup_{i \in \overline{I}}B_i \cup (S' \setminus \cup_{i\in \overline{I}}T(v_i))$ is a $(\Delta_0,|A|,m+1,\overline{m},T(v))-$BR-strategy due to Proposition \ref{prop:new3}.

From this we get that: 
\begin{equation*}U(u,S',\Delta_0,|A|,T(v)) =U(u,S'',\Delta_0,|A|,T(v))\end{equation*}

Therefore, it is enough if we prove that $U(u,S,\Delta_0,|A|,T(v))=U(u,S'',\Delta_0,|A|,T(v))$.

Now notice that the following results hold. On the one hand, we have that $U(u,S'' \cap T(v_j),\Delta_0,|A|,T(v_j)) = U(u,B_j,\Delta_0,|A|,T(v_j))$ for every $j\in \overline{I}$ because $S'' \cap T(v_j) = B_j$ and, by construction, there is no $(S,\Delta_0,T(v_j))-$blocking pair with $j \in J_=$ because there exists only one element from $S$ in $T(v_j)$ which is $v_j$. On the other hand, let $j \in J'$, where $J'$ is the subset of subindices $j$ from $I$ such that $S'' \cap T(v_j) \neq \emptyset$ i.e., $|S'' \cap T(v_j)| = 1$, because, for each $i \in I$ we know by Proposition \ref{prop:new2} that $|S'' \cap T(v_i)| \leq 1$. We claim that there does not exist any $(S'',\Delta_0,T(v_j))-$blocking pair with $j \in J'$ neither. This is because of the following reasoning: If $(a,b)$ is such a $(S'',\Delta_0,T(v_j))$-blocking pair, first we have that $\Delta(a,\emptyset) < \Delta(a,S'') = \Delta_0$, because  of the property $\Delta(z,Z) \leq \Delta(z,Z')$ whenever $Z \subseteq Z'$ together with the fact that $|S'' \cap T(v_j)| = 1$ by hypothesis. In second place, since $B_j=\emptyset$ and $B_j$ is a $(\Delta_0,T(v_j),m_j)$-strategy (in  fact $B_j$ is a $(\Delta_0,|A|, T(v_j),m_j)$-BR-strategy) then $\Delta(a,\emptyset)= \Delta(a,B_j) \geq \Delta_0$ for every $a\in \mathcal{U}[T(v_j)]$. Hence we have that   $ \Delta_0 \leq  \Delta(a,\emptyset) <  \Delta_0$, which is a contradiction. 


Therefore, combining these results together with Lemma \ref{lemm:ccblockingpair} we obtain that $CC_u(a,B_j) = CC_u(a,\emptyset)$ for every $a\in A(S \cap T(v_j),\Delta_0,T(v_j))$ with $j \in J_=$ and $CC_u(a,S'' \cap T(v_j)) = CC_u(a,\emptyset)$ for every $a \in A(S'' \cap T(v_j),\Delta_0,T(v_j))$ with $j \in J'$. Hence, for every $j \in J_=$:
\begin{equation*}U(u,B_j, \Delta_0,|A|,T(v_j)) = - \alpha + U(u,\emptyset,\Delta_0,|A|,T(v_j))\end{equation*} 
And, similarly, for every $j \in J'$:
\begin{equation*}U(u,S'' \cap T(v_j), \Delta_0,|A|,T(v_j)) = - \alpha + U(u,\emptyset,\Delta_0,|A|,T(v_j))\end{equation*} 
Furthermore, $U(u,S'' \cap T(v_j), \Delta_0,|A|,T(v_j)) = U(u,\emptyset, \Delta_0,|A|,T(v_j))$ for every $j \in I \setminus J'$ and $U(u,S \cap T(v_j), \Delta_0,|A|,T(v_j))=U(u,\emptyset, \Delta_0,|A|,T(v_j))$ for every $j \in I \setminus J_=$. From all these relationships we obtain:
\begin{equation*}\sum_{j=1}^{k(v)} U(u,S'' \cap T(v_j),\Delta_0,|A|,T(v_j)) =\sum_{j\in \overline{I}} U(u,S'' \cap T(v_j),\Delta_0,|A|,T(v_j)) +\end{equation*}
\begin{equation*}+\sum_{j\in J'} U(u,S'' \cap T(v_j),\Delta_0,|A|,T(v_j))+\sum_{j\in I\setminus J'} U(u,S'' \cap T(v_j),\Delta_0,|A|,T(v_j))=\end{equation*}
\begin{equation*}= \sum_{j\in \overline{I}} U(u,B_j,\Delta_0,|A|,T(v_j))+\sum_{j\in J'}\left(-\alpha+ U(u,\emptyset,\Delta_0,|A|,T(v_j))\right) +\end{equation*}
\begin{equation*}+\sum_{j\in I\setminus J'} U(u,\emptyset,\Delta_0,|A|,T(v_j)) = \sum_{j=1}^{k(v)} U(u,B_j,\Delta_0,|A|,T(v_j)) - \alpha |J'|\end{equation*}

And, similarly:
\begin{equation*}\sum_{j=1}^{k(v)} U(u,S \cap T(v_j),\Delta_0,|A|,T(v_j)) =  \sum_{j=1}^{k(v)} U(u,B_j,\Delta_0,|A|,T(v_j)) - \alpha |J_=|\end{equation*}

Now, using the  formula from Proposition \ref{lemm:restrictedutility}, we obtain: 
\begin{equation*}U(u,S'',\Delta_0,|A|,T(v)) = \end{equation*}
\begin{equation*}=\sum_{j=1}^{k(v)} U(u,S'' \cap T(v_j),\Delta_0,|A|,T(v_j)) + \frac{1}{|A|} \left(|\overline{T(v)}|+\sum_{j \in \overline{I} \cup J'} |T(v_j)| \right) = \end{equation*}
\begin{equation*}=   \sum_{j=1}^{k(v)} U(u,B_j,\Delta_0,|A|,T(v_j)) - \alpha |J'|+ \frac{1}{|A|} \left(|\overline{T(v)}|+\sum_{j \in \overline{I} \cup J'} |T(v_j)| \right)\end{equation*}

And, similarly:
\begin{equation*}U(u,S,\Delta_0,|A|,T(v)) = \end{equation*}
\begin{equation*}= \sum_{j=1}^{k(v)} U(u,B_j,\Delta_0,|A|,T(v_j)) - \alpha |J_=|+ \frac{1}{|A|} \left(|\overline{T(v)}|+\sum_{j \in \overline{I} \cup J_=} |T(v_j)| \right)\end{equation*}

Therefore, by construction of $S = S_=[B_1,...,B_{k(v)};m_1,...,m_{k(v)}]$,
\begin{equation*}U(u,S,\Delta_0,|A|,T(v)) -U(u,S'',\Delta_0,|A|,T(v)) = \end{equation*}
\begin{equation*}=-\alpha (|J_=|-|J'|)+\frac{1}{|A|}\left(\sum_{j \in J_=}|T(v_j)|-\sum_{j \in J'}|T(v_j)| \right) \geq 0\end{equation*}

Implying $U(u,S,\Delta_0,|A|,T(v))=U(u,S'',\Delta_0,|A|,T(v))$, as we wanted to see.
\end{proof}


\section{Computing a BR in Polynomial Time}

Finally, we have reached the last section in which we are ready to give the main algorithm to compute the best response. We use a dynamic programming approach exploiting all the main results of the previous sections. More precisely, the main technique in our algorithm has some reminiscence to the  knapsack-problem, although is considerably more complex and involved. In order to understand our algorithm we introduce some arrays that help us to obtain a final solution.  

\subsection{Definition of the arrays}

Let $M[v,\Delta_0,|A|,m]$ be the maximum $(\Delta_0,|A|,T(v))-$restricted utility a $(\Delta_0,m,T(v))-$strategy can attain or $-\infty$ if there does not exist any such $(\Delta_0,m,T(v))-$strategy. In case that $M[v,\Delta_0,|A|,m] \neq -\infty$ then $S[v,\Delta_0,|A|,m]$ will be such a $(\Delta_0,|A|,m,T(v))-$BR-restricted strategy.

\textbf{Auxiliary arrays.} Let $v\in V(T)$ be a node from the meta-tree. For each $i = 1,..., k(v)$ define $T^i(v) = \cup_{j \geq 1}^i T(v_j)$. We introduce an auxiliary array that can help us to compute the values $M[v,\Delta_0, |A|, m]$:

\vskip 5pt
\noindent \emph{(a) If $v \in V_1(G')$.} Let $M_{aux}[v,\Delta_0,|A|,m,i]$ be the maximum $(\Delta_0,|A|,T^i(v))-$restricted utility a $(\Delta_0,m,T^i(v))-$strategy can attain or $-\infty$ if there does not exist any such $(\Delta_0,m,T^i(v))-$strategy. In case that $M_{aux}[v,\Delta_0,|A|,m,i] \neq \infty$ then $S_{aux}[v,\Delta_0,|A|,m,i]$ will be such a $(\Delta_0,|A|,m,T^i(v))-$BR-restricted strategy.

\vskip 10pt

\noindent \emph{(b) If $v \in V_{\geq 2}(G')$}. Let $M_{aux}[v,\Delta_0,|A|,m,i,CC_u,CC_{u,2},JJ]$ be the array with value in every position to be the maximum $(\Delta_0,|A|,T^i(v))-$restricted utility a $(\Delta_0,m,T^{i}(v))-$strategy $S \subseteq T^i(v)$  can attain given that

(i) $|CC_u(v,S)| = |\overline{T(v)}|+\sum_{S \cap T(v_j) \neq \emptyset} |T(v_j)| = CC_u$

(ii) $|\overline{T(v)}|^2+\sum_{S \cap T(v_j) \neq \emptyset}|T(v_j)|^2 = CC_{u,2}$

(iii) $-\alpha |J(S)| + \frac{1}{|A|} \sum_{j \in J(S) }|T(v_j)| = JJ$ where  $J(S)$ is the subset of indices $j$ with $1\leq j \leq i$ satisfying $S \cap T(v_j) = \left\{v_j\right\}$. 

Or $-\infty$ if such a $(\Delta_0,m,T^i(v))-$strategy does not exist. In case that $M_{aux}[v,\Delta_0,|A|,m,i,CC_u,CC_{u,2},JJ] \neq -\infty$ then $S_{aux}[v,\Delta_0,|A|,m,i,CC_u,CC_{u,2},JJ]$ will be such a $(\Delta_0,|A|,T^i(v),m)-$BR-restricted strategy.

\subsection{Recurrence relations}

We have seen in the previous sections how a restricted BR-strategy for a node $v\in V(T)$ can be obtained in terms of the restricted BR-strategies for the corresponding subtrees $T(v_1),...,T(v_{k(v)})$. This is what we are going to exploit in the next results.

\textbf{The arrays $M,S$ in terms of the arrays $M_{aux},S_{aux}$.} 

First, we see how can we compute the corresponding values from $M,S$ assuming that we have previously computed the corresponding values from $M_{aux}, S_{aux}$, and distinguishing between the two cases $v\in V_1(G')$ and $v\in V_{\geq 2}(G')$. 

These results follow easily from the definitions of the arrays:

\noindent \emph{(a) Suppose that $v \in V_{1}(G')$.}  Then:
$$M[v,\Delta_0,|A|,m] = M_{aux}[v, \Delta_0, |A|, m - \widetilde{m}(v,\Delta_0),k(v)]$$
And if $M[v,\Delta_0,|A|,m] \neq -\infty$ then:
$$S[v,\Delta_0,|A|,m] = S_{aux}[v, \Delta_0, |A|, m - \widetilde{m}(v,\Delta_0),k(v)]$$

\noindent \emph{(b) Suppose that $v\in V_{\geq 2}(G')$.} Then notice first of all the next result. 

For any strategy $S$: 
$$\Delta(v,S) = \Delta(v,\emptyset)- \left( |\overline{T(v)}|^2 + \sum_{S \cap T(v_j) \neq \emptyset} |T(v_j)|^2\right)+|CC_u(v,S)|^2$$

Therefore, given the parameters $\Delta_0$, $CC_{u,2}$ and $CC_u$ the next value indicates whether $v$ is attacked in the BR-strategy we are building (or $-\infty$ if there is no such BR-strategy):

$$\epsilon(v,\Delta_0,CC_u,CC_{u,2}) = 
\begin{cases} 
0   & \text{if $\Delta_0 < \Delta(v,\emptyset)-CC_{u,2}+CC_u^2$} \\ 
1   & \text{if $\Delta_0 = \Delta(v,\emptyset)-CC_{u,2}+CC_u^2$} \\ 
-\infty & \text{ otherwise}\end{cases} $$

Then:
$$M[v,\Delta_0,|A|,m] =$$
$$= \max_{\substack{CC_u, CC_{u,2} ,JJ \\ \epsilon(v,\Delta_0,CC_u,CC_{u,2}) \neq -\infty} }M_{aux}[v, \Delta_0, |A|, m-\epsilon(v,\Delta_0,CC_u,CC_{u,2})  ,k(v),CC_u,CC_{u,2}, JJ]$$
If $M_{aux}[v,\Delta_0,|A|,m] \neq -\infty$ consider 
$$CC_u^*,CC_{u,2}^*,JJ^*=$$
$$= \arg\max_{\substack{CC_u, CC_{u,2} ,JJ \\ \epsilon(v,\Delta_0,CC_u,CC_{u,2}) \neq -\infty} }M_{aux}[v, \Delta_0, |A|, m-\epsilon(v,\Delta_0,CC_u,CC_{u,2})  ,k(v),CC_u,CC_{u,2}, JJ]$$
Finally we  we set $S[v,\Delta_0,|A|,m] =  S_{aux}[v,\Delta_0,|A|,m,i,CC_u^*,CC_{u,2}^*,JJ^*]$.

\textbf{The arrays $M_{aux},S_{aux}$ in terms of the arrays $M,S,M_{aux},S_{aux}$.}  

Here comes a trickier part. As before, we need to distinguish between the two scenarios $v\in V_1(G')$ and $v\in V_{\geq 2}(G')$.

\noindent \emph{(a) Suppose that $v\in V_{1}(G')$.}  

In this scenario we distinguish the main recurrence (the case $i>0$) and the corresponding initialization of the array (the case $i=0$).

\begin{itemize}

\item If $i > 0$. We know by Theorem \ref{THM:1} that by the disjoint union of the restricted BR-strategies from $T(v_i)$ with $1 \leq i \leq k(v)$ we obtain a restrict BR-strategy for $T(v)$ with no need of buying more links. In view of this we set:
 
$$M_{aux}[v,\Delta_0,|A|,m,i] = \max_{0 \leq m_i \leq m} \left\{ M_{aux}[v,\Delta_0,|A|,m-m_i,i-1]+M[v_i,\Delta_0,|A|,m_i] \right\}$$
And if $M_{aux}[v,\Delta_0,|A|,m,i] \neq -\infty$ consider 
$$m_i^* = \argmax_{0 \leq m_i \leq m} \left\{M_{aux}[v,\Delta_0,|A|,m-m_i,i-1]+M[v_i,\Delta_0,|A|,m_i] \right\}$$
Then in such case we can set $S_{aux}[v,\Delta_0,|A|,m,i] =  S_{aux}[v,\Delta_0,|A|,m-m_i^{*},i-1]\cup S[v_i,\Delta_0,|A|,m_i^{*}]$. 

\item  Otherwise, if $i=0$, then $M_{aux}[v,\Delta_0,|A|,m,i] = -\infty$ if $m \neq 0$ and $M_{aux}[v,\Delta_0,|A|,m,i] = 0$ and $S_{aux}[v,\Delta_0,|A|,m_i]=\emptyset$ otherwise. 

\end{itemize}

\noindent \emph{(b) Suppose that $v\in V_{\geq 2}(G')$.} 

As before, we distinguish the main recurrence (the case $i>0$) and the corresponding initialization of the array (the case $i=0$). 

\begin{itemize}

\item If $i > 0$, let $B_i$ be the restricted BR-strategy corresponding to the subtree $T(v_i)$.  We distinguish between three situations: 

Scenario 1: if $B_i \neq \emptyset$, then $u$ does not need to buy any link to $T(v_i)$ using Theorem \ref{thm:2}. 

Scenario 2: if $B_i = \emptyset$. Then, again by Theorem \ref{thm:2} $u$ buys at most one link to $B_i$ so there are two possible cases:

\hspace{1cm} Case 2.1: $u$ buys one link towards $\left\{v_i\right\}$.

\hspace{1cm} Case 2.2: $u$ buys no link towards $T(v_i)$. 
Therefore:

$$M_{aux}[v,\Delta_0,|A|,m,i,CC_u,CC_{u,2},JJ] =\max(M_1,M_2,M_3)$$

Where $M_1,M_2,M_3$ correspond to scenario 1 and scenario 2 cases 2.1 and 2.2, respectively:
$$M_1= \max_{0 \leq m_i \leq m} \{ M_{aux}[v,\Delta_0,|A|,m-m_i,i-1,CC_u-|T(v_i)|,CC_{u,2}-|T(v_i)|^2,JJ] +$$
$$+M[v_i,\Delta_0,|A|,m_i] +\frac{1}{|A|}|T(v_i)| \}$$
$$M_2 = \max_{0 \leq m_i \leq m} \{ M_{aux}[v,\Delta_0,|A|,m-m_i,i-1,CC_u-|T(v_i)|,CC_{u,2}-|T(v_i)|^2,JJ-\alpha+ \frac{|T(v_i)|}{|A|}]+$$
$$+M[v_i,\Delta_0,|A|,m_i]+\frac{1}{|A|}|T(v_i)| \}$$
$$M_3= \max_{0 \leq m_i \leq m} \{M_{aux}[v,\Delta_0,|A|,m-m_i,i-1,CC_u,CC_{u,2},JJ] + M[v_i,\Delta_0,|A|,m_i]\}$$

If $M_1 \neq -\infty$ let 
$$m_{i,1}^* = \argmax_{0 \leq m_i \leq m} \{M_{aux}[v,\Delta_0,|A|,m-m_i,i-1,CC_u-|T(v_i)|,CC_{u,2}-|T(v_i)|^2,JJ] +$$
$$+M[v_i,\Delta_0,|A|,m_i] +\frac{1}{|A|}|T(v_i)|\}$$
Then let  $S_1 = S_{aux}[v,\Delta_0,|A|,m-m_{i,1}^*,i-1,CC_u-|T(v_i)|,CC_{u,2}-|T(v_i)|^2,JJ] \cup S[v_i,\Delta_0,|A|,m_{i,1}^*]$.

If $M_2 \neq -\infty$ let 
$$m_{i,2}^* = \argmax_{0 \leq m_i \leq m} \left\{M_{aux}[v,\Delta_0,|A|,m-m_i,i-1]+M[v_i,\Delta_0,|A|,m_i] \right\}$$
Then let  $S_2 = \left\{v_i\right\} \cup S_{aux}[v,\Delta_0,|A|,m-m_{i,2}^*,i-1,CC_u-|T(v_i)|,CC_{u,2}-|T(v_i)|^2,JJ] \cup S[v_i,\Delta_0,|A|,m_{i,2}^*]$.

If $M_3 \neq -\infty$ let 
$$m_{i,3}^* = \argmax_{0 \leq m_i \leq m} \{ M_{aux}[v,\Delta_0,|A|,m-m_i,i-1,CC_u,CC_{u,2},JJ] + M[v_i,\Delta_0,|A|,m_i]\}$$
Then let  $S_3 = S_{aux}[v,\Delta_0,|A|,m-m_{i,3}^*,i-1,CC_u,CC_{u,2},JJ] \cup S[v_i,\Delta_0,|A|,m_{i,3}^*]$.

Then if $\max(M_1,M_2,M_3) \neq -\infty$ we set $S_{aux}[v,\Delta_0,|A|,m,i] =  S_{j^*}$ where $j^* =  \argmax_{ 1 \leq j \leq 3} \{M_j\}$.

\item Otherwise, if $i = 0$, then:
$$M_{aux}[v,\Delta_0,|A|,m,0,CC_u,CC_{u,2},JJ]=-\infty$$
If $m \neq 0$, $CC_u \neq 0$, $CC_{u,2} \neq 0$ or $JJ \neq 0$. Otherwise:
$$M_{aux}[v,\Delta_0,|A|,0,0,0,0,0]=0$$
$$S_{aux}[v,\Delta_0,|A|,0,0,0,0,0]= \emptyset$$
\end{itemize}

\subsection{Final recurrence}

Let 
$$\Delta_1^*, |A|_1^* = \argmax_{  \Delta_0, |A|} M[u,\Delta_0,|A|,|A|]$$ 
And
$$\Delta_2^*,|A|_2^* = \argmax_{ \Delta_0, |A| }M[u,\Delta_0,|A|+1,|A|]$$

Then, by Proposition \ref{prop:restricted} we can pick as a best response $(S[u, \Delta_1^*,|A|_1^*,|A|_1^*] ,1)$ if  $M[u,\Delta_1*,|A|_1^*,|A|_1^*]-\beta > M[u,\Delta_2^*,|A|_2^*+1,|A|_2^*]$ or $(S[u, \Delta_2^*,|A|_2^*+1,|A|_2^*] ,0)$ otherwise. 

Furthermore, notice that there are a polynomial (in $n$) number of possible values taking at most polynomial values (again, with respect $n$) for the distinct parameters of the array, then our algorithm can be computed in polynomial time, which is what the conjecture in this particular case claims.

Even though our algorithm solves the best response problem considering that the initial graph resulting from the strategies of all players but $u$, is connected, we believe that our techniques can be adapted to solve the problem in the general case.

\end{document}